\documentclass[onecolumn, twoside]{IEEEtran}
\usepackage{amsmath, amsthm, amssymb}
\usepackage{latexsym}
\usepackage{graphicx}
\usepackage{verbatim}
\usepackage{mathrsfs}
\usepackage{float}
\usepackage[lofdepth, lotdepth]{subfig}
\usepackage{array}
\usepackage{url}
\usepackage{algorithm}
\usepackage[noend]{algorithmic}
\usepackage[table]{xcolor}
\usepackage{enumerate}
\usepackage{eucal}
\usepackage{stmaryrd}
\usepackage{mathtools}
\usepackage{pgf}
\usepackage{tikz}
\usetikzlibrary{arrows,automata,positioning}
\usepackage[latin1]{inputenc}
\usetikzlibrary{automata}

\usepackage[top=0.7in, bottom=0.9in, left=0.7in, right=0.7in]{geometry}
\usepackage{mathtools}
\newcommand{\floor}[1]{\left\lfloor #1 \right\rfloor}
\newcommand{\ceiling}[1]{\left\lceil #1 \right\rceil}

\usepackage{setspace}
\newcolumntype{C}[1]{>{\centering\arraybackslash}p{#1}}

\theoremstyle{plain}
\newtheorem{theorem}{Theorem}
\newtheorem{proposition}[theorem]{Proposition}
\newtheorem{lemma}[theorem]{Lemma}

\newtheorem{corollary}[theorem]{Corollary}
\newtheorem{claim}{Claim}

\theoremstyle{definition}
\newtheorem{definition}{Definition}
\newtheorem{example}[definition]{Example}
\newtheorem{remark}[definition]{Remark}

\makeatletter
\let\@@pmod\pmod
\DeclareRobustCommand{\pmod}{\@ifstar\@pmods\@@pmod}
\def\@pmods#1{\mkern4mu({\operator@font mod}\mkern 6mu#1)}
\makeatother

\newcommand{\B}{{\mathcal B}}
\newcommand{\C}{{\mathcal C}}

\newcommand{\cE}{{\mathcal E}}

\DeclareMathAlphabet{\mathbfsl}{OT1}{ppl}{b}{it} %{OT1}{cmr}{bx}{it}

\newcommand{\bF}{{\mathbfsl F}}
\newcommand{\bH}{{\mathbfsl H}}

\newcommand{\bL}{\mathbfsl{L}}

\newcommand{\bB}{{\mathbfsl B}}

\newcommand{\bp}{{\mathbfsl p}}
\newcommand{\bu}{{\mathbfsl u}}
\newcommand{\bv}{{\mathbfsl v}}

\newcommand{\by}{{\mathbfsl y}}

\newcommand{\bc}{{\mathbfsl c}}

\newcommand{\bh}{{\mathbfsl{h}}}

\newcommand{\bx}{{\mathbfsl{x}}}
\newcommand{\bz}{{\mathbfsl{z}}}

\newcommand{\bm}{{\mathbfsl{m}}}
\newcommand{\sq}{{\sf q}}
\newcommand{\ra}{\rightarrow}

\newcommand{\weight}{{\mathsf{wt}}}
\newcommand{\level}{{\mathsf{level}}}
\newcommand{\synd}{{\mathsf{syn}}}
\newcommand{\balind}{{\mathsf{index}}}

\newcommand{\flip}{{\mathsf{Flip}}}

\renewcommand{\geq}{\geqslant}
\renewcommand{\leq}{\leqslant}
\renewcommand{\ge}{\geqslant}
\renewcommand{\le}{\leqslant}

\newcommand{\enc}{\textsc{Enc}}

\newcommand{\bbracket}[1]{\left \llbracket #1 \right\rrbracket}
\newcommand{\shift}{{\boldsymbol{\sigma}}}

\newcommand{\todo}[1]{{\color{red}(TODO: #1)}}

\newcommand{\etal}{{\em et al.}}

\begin{document}

%\title{Average Redundancy of \\Variable-Length Knuth-Like Balancing Schemes}
\title{Average Redundancy of\\Variable-Length Balancing Schemes \`{a} la Knuth}
%\title{Average\,Redundancy\,of\,Variable-Length\,Balancing\,Schemes}

%	\title{Variable-Length Knuth Balancing Method for $\sq$-balanced code\\[-4mm]}
	\author{
		\IEEEauthorblockN{
			Duc Tu Dao\IEEEauthorrefmark{1},
			Han Mao Kiah\IEEEauthorrefmark{1}, and
			Tuan Thanh Nguyen\IEEEauthorrefmark{2}}
		
		\IEEEauthorblockA{
			\IEEEauthorrefmark{1}\small School of Physical and Mathematical Sciences, Nanyang Technological University, Singapore\\
			\IEEEauthorrefmark{2}\small Singapore University of Technology and Design\\
			email: \{{ductu.dao}, {hmkiah}\}@ntu.edu.sg, tuanthanhnguyen@sutd.edu.sg
		}
		\thanks{Results in Section \ref{sec:schemeA},~\ref{sec:lattice}, and \ref{sec:schemeB} have been presented at  International Symposium on Information Theory and Its Applications 2022~\cite{ISITA22} }
			%\todo{HM: Insert citation with details of the conference paper}.}
	}

	\maketitle

	%SECTION ABSTRACT 
	\hspace{-3mm}\begin{abstract}
	We study and propose schemes that map messages onto constant-weight codewords using variable-length prefixes. We provide polynomial-time computable formulas that estimate the average number of redundant bits incurred by our schemes. 
	In addition to the exact formulas, we also perform an asymptotic analysis and demonstrate that our scheme uses $\frac12 \log n+O(1)$ redundant bits to encode messages into length-$n$ words with weight $(n/2)+{\sf q}$ for constant ${\sf q}$.
	We also propose schemes that map messages into balanced codebooks with error-correcting capabilities. 
	For such schemes, we provide methods to enumerate the average number of redundant bits.
	\end{abstract}
	
	%%%%%%%%%%
\begin{comment}
	{
		\begin{itemize}
			\item Include reference to our paper is presented in ISITA. Please see if you have a copy of the paper in the proceedings. Hopefully there are page numbers on it.
			\item Section VI-D: Check Fig~\ref{fig:743cyclic} again.
			\item Section VII: Help to write the conclusion. Mainly focus on writing open problems.
			\item Help to draft a cover letter. I shared an example in the folder. In the folder, I also attach a draft of the response document (this is for the future).
			%\item Standardize $\leq$ to $\le$.
	\end{itemize}}
	
\end{comment}
	
\section{Introduction}
	
The {\em imbalance} of a binary word $\bx$ refers to the difference%
\footnote{For ease of exposition in later sections, we use the term {\em imbalance} to refer to the (signed) difference, instead of the absolute difference. In other words, a positive imbalance indicates that there are more ones than zeroes.} between the number of ones and the number of zeros in $\bx$. 
A~word of even length is {\em balanced} if its imbalance is exactly zero and a code is {\em balanced} if all its codewords are balanced. 
Due to their applications in various recording systems, balanced codes have been extensively studied (see \cite{ImminkText} for a survey).
In recent years, interest in balanced codes has been rekindled because of 
the emergence of DNA macromolecules as a next-generation data storage medium with its unprecedented density, durability, and replication efficiency \cite{Shomorony2022, Yazdi2015}. 
Specifically, a DNA string comprises four bases or letters: {\tt A}, {\tt T}, {\tt C}, and {\tt G}, and a string is {\tt GC}-rich (or {\tt GC}-poor) if a high (or
low) proportion of the bases corresponds to either {\tt G} or {\tt C}. Since
{\tt GC}-rich or {\tt GC}-poor DNA strings are prone to both synthesis and
sequencing errors \cite{ImminkCai2020, Ross2013}, we aim to reduce the difference between
the number of {\tt G} and {\tt C} and the number of {\tt A} and {\tt T} on every DNA codeword. This requirement turns out to be equivalent to reducing the imbalance of a related binary word (see for example \cite{CaiChee2021, ImminkCai2020}).

In his seminal paper \cite{knuth1986efficient}, Knuth proposed a simple and elegant linear-time algorithm that transforms an arbitrary binary length-$n$ message into a balanced length-$n$ codeword.
To allow the receiver to recover the message, a $\ceiling{\log_2 n}$-bit {\em prefix} must be transmitted and hence, Knuth's method incurs a redundancy of $\ceiling{\log_2 n}$ bits. This differs from the minimum required by a multiplicative factor of two (see \eqref{eq:optimal-red}). 
Later, Alon \etal~\cite{alon1988balancing} demonstrated that, under certain assumptions of the encoding scheme,  $\log_2 n$ redundant bits are necessary (see discussion in Section~\ref{sec:prelim} for more details). 
Hence, it appears unlikely that we can improve Knuth's balancing technique if we insist on transmitting prefixes of a {\em fixed length}.

Therefore, in~\cite{immink2010very, Weber2010}, Immink and Weber proposed balancing schemes that transmit {\em variable-length prefixes} and studied the average redundancy of their proposals. 
Specifically, in~\cite{Weber2010}, Weber and Immink provided two variable-length balancing schemes whose average redundancy are asymptotically equal to $\log_2 n$ and $\frac12\log_2 n+0.916$, respectively. 
The work in \cite{Weber2010} was later extended to $q$-ary case in~\cite{Paluncic2018Capacity},
where the schemes has average redundancy $\frac12\log_q n+C_q$ for some constant $q$. 
%, {\color{magenta}where the authors explain why applying auxiliary data is able to recover nearly $\frac{1}{2}\log_q{n}$ redundancy loss in Knuth balancing schemes} \todo{HM: What does this line mean? Specifically, what does explaining mean? Did they have a proof?}.

Now, in a separate work~\cite{immink2010very}, Immink and Weber proposed another variable-length balancing scheme which we study closely in this paper. 
In~\cite{immink2010very}, Immink and Weber provided closed formulas for the average redundancy of their scheme and computed these values for $n\le 8192$. 
While numerically the redundancy values are close to the optimal value given in~\eqref{eq:optimal-red}, a tight asymptotic analysis was not provided. A slightly different approach was considered by Swart and Weber in \cite{ Swart2018binary}, where messages have variable length and codewords have fixed length. Although its encoding/decoding is simple, its redundancy increases to $O(\sqrt{n})$ for long messages.
In this work, we make modifications to the scheme in~\cite{immink2010very} and demonstrate that the average redundancy of the resulting scheme\footnote{Specifically, Scheme~B with $\sq=0$.}
 is at most $\frac12\log_2 n+0.526$ asymptotically.

Even though the average redundancy of our scheme differs from the optimal \eqref{eq:optimal-red} by an additive constant of approximately $0.2$, our scheme and its accompanying analysis can be easily extended to the case where the imbalance is fixed to some positive constant.
Formally, for an even integer $n$ and some fixed integer $\sq$,  we say that a length-$n$ word is {\em $\sq$-balanced} if its imbalance is exactly $\sq$, that is, its weight is exactly $(n/2)+\sq$. 
A code is {\em $\sq$-balanced} if all words are $\sq$-balanced.
Since all words in an $\sq$-balanced have the same weight, such codes are also known as {\em constant-weight codes} and are used in a variety of communication and data storage scenarios.
Recent applications involve data storage in crossbar resistive memory arrays \cite{NguyenVuCai2021, Nguyen2021} and live DNA \cite{Chen2021} (see also \cite{Tallini1998} for a survey).

While there is extensive research on constructing constant-weight codes with distance properties, simple efficient encoding methods are less well-known. In fact, this problem was posed by MacWilliams and Sloane as {\bf Research Problem~17.3}~\cite{MS1977}. To the best of our knowledge, there are three encoding approaches: the enumerative method of Schalkwijk~\cite{Schalkwijk1972}, the geometric approach of Tian \etal{}~\cite{Tian2009}, and the Knuth-like method of Skachek and Immink \cite{skachek2014constant}. The Knuth-like method is also viable to construct $q$-ary constant weight sequences with redundancy $\log_q{n}+O(1)$ \cite{Mambou2018Capacity}.
For the case where $\sq$ is constant, the first two methods encode in quadratic time $O(n^2)$, while the third method runs in linear time.
However, the third method incurs $\log_2 n$ redundant bits and this is the regime that we study in this work.
Specifically, when $\sq$ is a positive constant, we show that there is a linear-time variable-length $\sq$-balancing scheme that incurs average redundancy of at most $\frac12 \log_2 n+2.526$ redundant bits.

We also propose variable-length balancing schemes that are equipped with error-correcting capabilities.
This is pertinent for variable-length encoding schemes as errors in prefixes may lead to catastrophic error propagation.
Previous constructions for error-correcting balanced codes were given in \cite{Albassam1990, Blaum1989, chee2020, Kumar2023, Mazumdar2009, Tilborg1989, Weber2010}. 
We highlight two techniques that transform general linear error-correcting codes into balanced codes with similar error-correcting capabilities. 
In~\cite{Weber2010},  Weber~\etal{} take two error-correcting codes as inputs: an $(n, d)$-code $\C_n$ and a short $(p,d)$-balanced code $\C_p$, and construct a $(n+p,d)$-balanced code with minimum distance $d$ and has the same size as $\C_n$.
Hence, their method introduced $p$ redundant bits and the distance is necessarily bounded by $p$.
Later, in~\cite{chee2020}, Chee~\etal{} removed the need of the short balanced codebook $\C_p$.
Specifically, using a $(n,d)$-cyclic code of size $M$, they constructed a $(n,d)$-balanced code of size at least $M/n$. 
Of significance, the distance of the resulting balanced code can be as large as $n$.
To achieve this, Chee~\etal{} introduced a variant of Knuth's balancing technique, which we term as {\em cyclic balancing} (see Section~\ref{sec:ecc}). Similar to previous schemes, we propose a variable-length balancing scheme for this cyclic balancing technique and analyze the corresponding average redundancy.

\begin{comment}
In \cite{Tilborg1989}, van Tilborg and Blaum introduced the idea to consider short balanced blocks as symbols of an alphabet and construct error-correcting codes over that alphabet. Only moderate rates can be achieved by this method, but it has the advantage of limiting the digital sum variation and the run lengths. 
In \cite{Albassam1990}, Al-Bassam and Bose constructed balanced codes correcting a single error, which can be extended to codes correcting up to two, three, or four errors by concatenation techniques. 
In \cite{Mazumdar2009}, Mazumdar, Roth, and Vontobel considered linear balancing sets and applied such sets to obtain error-correcting coding schemes in which the codewords are balanced.

Recently, Weber et al. \cite{Weber2010} modified Knuth's balancing technique to endow the code with error-correcting capabilities. Their method requires two error-correcting codes as inputs: an $(m, d)_2$ code $C_m$ and a short $(p,d)_2$ balanced code $C_p$ where $|C_p| \geq m$. Given a message, they first encode it into a codeword $m \in C_m$. Then they find the balancing index z of m and flip the first z bits to obtain a balanced c. Using $C_p$, they encode z into a balanced word p and the resulting codeword is $c_p$. Since both $C_m$ and $C_p$ have distance $d$, the resulting code has minimum distance $d$. The method requires the existence of a short balanced code $C_p$. 
\end{comment}

In summary, our contributions are as follow.
First, we adapt the variable-length balancing schemes in~\cite{immink2010very} to map messages into $\sq$-balanced words for a fixed $\sq\ge 0$.
Second, and more crucially, we provide a detailed analysis of the average redundancy of our variable-length $\sq$-balancing schemes.
To this end, we borrow tools from lattice-path combinatorics and provide closed formulas for the upper bounds on the average redundancy of both Schemes A and B (described in Sections~\ref{sec:schemeA} and~\ref{sec:schemeB}). 
Unfortunately, as with~\cite{immink2010very}, we are unable to complete the asymptotic analysis for Scheme~A. 
Hence, we introduce Scheme~B which uses slightly more redundant bits, and
show that Scheme B incurs average redundancy of at most $\frac12 \log_2 n+2.526$ redundant bits asymptotically when $\sq>0$. 
Interestingly, for the case $\sq=0$, the average redundancy of Scheme B can be reduced to $\frac12 \log_2 n+0.526$ and this is better than the schemes given in~\cite{Weber2010}.
Finally, when $\sq=0$, we introduce Scheme~C that has certain error-correcting capabilities.
We then propose both closed formula and efficient methods to determine the average redundancy of Scheme~C.

In the next section, we formally define a variable-length $\sq$-balancing scheme and state our results.

	\section{Preliminaries}\label{sec:prelim}
	
	Throughout this paper, we fix $n$ to be an even integer and 
	let $\bbracket{n}$ denote the set of $n+1$ integers $\{0,1,2,\ldots,n\}$.
	For a binary word $\bx=x_1x_2\cdots x_n\in\{0,1\}^n$ and $j\in\bbracket{n}$, 
	we let $\bx^{[j]}$ and $\bx_{[j]}$ denote its length-$j$ prefix and length-$j$ suffix, respectively. 
	In other words, $\bx^{[j]}=x_1x_2\cdots x_j$ and $\bx_{[j]}=x_{n-j+1}x_{n-j+2}\cdots x_n$. 
	Also, for two binary words $\bx$ and $\by$, we use $\bx\by$ to denote their concatenation and $\bx^a$ to denote the concatenation of $a$ copies of $\bx$.
	Finally, we use $\overline{\bx}$ to denote the {\em complement} of $\bx$.
	
	Now, since $n$ is even, we set $n=2m$ and fix a non-negative integer $\sq \le m = n/2$. Recall that a length-$n$ word $\bx$ is {\em $\sq$-balanced} if its weight is exactly $m + \sq = n/2+\sq$. In other words, the imbalance of $\bx$ is exactly $2\sq$. 	
	The collection of all $\sq$-balanced words of length $n$ is denoted by $\B(n,\sq)$. If $\sq = 0$, then we simply write $\B(n,\sq)$ as $\B(n)$ and refer to these words as \textit{balanced} words.
%	A code $\C$ is $\sq$-\textit{balanced} or \textit{balanced} if all codewords in $\C$ are $\sq$-\textit{balanced} or \textit{balanced}, respectively.
	
	Here, our goal is to efficiently map arbitrary binary messages into codewords in collection $\B(n,\sq)$, while incurring as few redundant bits as possible. 
	Formally, we have the following.
	
	\begin{definition}
	An $(n,k,\rho;\sq)$-$\textit{variable-length balancing scheme}$ is a pair of encoding and decoding maps $(\mathsf{E},\mathsf{F})$ such that:
	%the following hold.
	\begin{enumerate}[(i)]
	\item  $\mathsf{E}$ is an injective {\em encoding} map from $\{0,1\}^k$ to $\B(n,\sq) \times \{0,1\}^*$. In other words, $\mathsf{E}(\bx)=(\bc,\bp)$ with $\bc\in\B(n,\sq)$ and $\bp\in \{0,1\}^*$, and we refer to $\bp$ as the {\em prefix}. Here, $\{0,1\}^*$ denotes the set of all finite-length binary words. 
	\item  $\mathsf{D}$ is a {\em decoding} map from ${\rm Im}(\mathsf{E})$ to $\{0,1\}^k$ such that $\mathsf{D} \circ \mathsf{E} (\bx) = \bx$ for all $\bx \in \{0,1\}^k$.
	\end{enumerate}
	Here, $\rho$ denotes the \textit{average redundancy} and it is given by 
	\begin{equation}
		\rho \triangleq \frac{1}{2^k}\sum_{\bx \in \{0,1\}^k} \big(|\mathsf{E}
		(\bx)| -k\big).
	\end{equation}
\end{definition}
	
	Clearly,  the cardinality of $\B(n,\sq)$ is given by $\binom{n}{n/2+\sq} = \frac{2^n e^{-2\sq^2/n}}{\sqrt{\pi n/2}} \Big(1+O\left(\frac1n\right)\Big)$, 
	%\todo{Tu, please check this!}, 
	where the latter asymptotic estimate follows from Stirling's approximation (see for example, \cite[Theorem~4.6]{sedgewick2013introduction}). Therefore, the redundancy of $\B(n,\sq)$, and thus, the minimum redundancy, is given by
	\begin{align}
	&n - \log_2	\binom{n}{n/2+\sq} \label{eq:optimal-exact}\\
	&	= \frac{1}{2}\log_2 n + \frac{1}{2}\log_2 \frac{\pi}{2} + \frac{2\sq^2}{n}\log_2 e + O\left(\frac1n\right) \notag\\
	& = \frac{1}{2}\log_2 n + \frac{1}{2}\log_2 \frac{\pi}{2} +O\left(\frac1n\right) \notag \\
	& \sim 
	 \frac{1}{2}\log_2 n + \Delta, \text{ where }\Delta\approx 0.326\ldots \,. \label{eq:optimal-red}
	\end{align}
	Here, we assume that $\sq$ is constant and asymptotics are taken with respect to $n$. That is, we have that $f(n)\sim g(n)$ if and only if $\lim_{n\to\infty} f(n)/g(n)=1$.
	
	As mentioned earlier, when $\sq=0$, we have the celebrated Knuth's balancing technique~\cite{knuth1986efficient}. 
	A crucial ingredient to this technique is the following simple flipping
	operation.
	For any word $\bx \in \{0,1\}^n$ and $j \in \bbracket{n}$, we define $\flip(\bx,j)$ to be the binary word obtained by flipping the first $j$ bits of $\bx$.
	That is, $\flip(\bx,j) = \overline{x_1 x_2 \cdots x_j}x_{j+1}\cdots x_n$.
	Then Knuth's technique simply searches for an index $j$ such that $\flip(\bx,j)$ belongs to $\B(n,\sq)$. Formally, we say that $j$ is a \textit{$\sq$-balancing index} for $\bx$ if $\flip(\bx,j)$ belongs to $\B(n,\sq)$ and we use $T(\bx,\sq)$ to denote the set of all $\sq$-balancing indices of $\bx$. 
	In other words, $T(\bx,\sq)=\{j \in \bbracket{n}: \flip(\bx,j) \in \B(n,\sq)\}$. 
	Knuth's key observation is that $T(\bx,0)$ is always nonempty~\cite{knuth1986efficient}, or equivalently, a $0$-balancing index always exists! 
	
	Let $\flip(\bx,\tau)=\bc$.
	In order for the receiver to recover the message $\bx$, the sender needs to transmit both $\bc$ and some representation of $\tau$. As the set of all possible balancing indices\footnote{It can be shown that there is a balancing index $\tau$ in the interval $0\le\tau\le n-1$.}
	has size $n$, we require a $\ceiling{\log_2 n}$-bit prefix $\bp$ to represent the index~$\tau$. 
	Hence, Knuth's balancing technique results in an $(n,n,\rho;0)$-variable-length balancing scheme with $\rho=\ceiling{\log_2 n}$. Note that this is in fact a fixed-length scheme because the prefix $\bp$ is always of length $\rho$. However, $\rho$ is twice the optimal quantity given in~\eqref{eq:optimal-red}. 
	One way to reduce this redundancy is to use a different and possibly smaller set of possible balancing words. Unfortunately, Alon~\etal~\cite{alon1988balancing}  demonstrated that the size of any balancing set must be at least $n$, and so $\log_2 n$ redundant bits are necessary. 
	So, in~\cite{immink2010very, Weber2010}, Immink and Weber turn their attention to variable-length balancing schemes. 
	We summarize their results here.
	
	\begin{theorem}[\cite{Weber2010, immink2010very}]\label{thm:immink}
	Let $\ell\in\{1,2,3\}$.
	There exists explicit $(n,n,\rho_\ell;0)$-variable-length balancing schemes with average redundancies\footnote{For purposes of brevity, we omit the closed formulas for Schemes 1 and 2.} as follows:
	\begin{align}
	\rho_1 & \sim \log n\,,\\
	\rho_2 & \sim \frac12 \log n+0.916\ldots\,,\\
	\rho_3 & = \frac{1}{2^{n}}\sum_{i=1}^{n/2} i\gamma(i,n)\log i \,,
	\end{align}where
	{ 
	\begin{equation}\label{eq:gamma}
		\gamma(i,n) = 2^n\left(\sum_{j=1}^i \cos^n\frac{\pi j}{i+1} - 2 \sum_{j=1}^{i-1} \cos^n\frac{\pi j}{i} + \sum_{j=1}^{i-2}\cos^n\frac{\pi j}{i-1}\right)\,.
	\end{equation}
	}
	We refer to these schemes as Schemes 1, 2, and 3, respectively.
	\end{theorem}

	When $\sq>0$, the set of $\sq$-balancing indices may be empty.
	That is, there exist binary words $\bx$ with $T(\bx,\sq)=\varnothing$ and we refer to such words as {\em bad} words.
	Hence, a different encoding rule must be applied to these bad words, and 
	simple linear-time methods were proposed and studied by Skachek and Immink~\cite{skachek2014constant}. 
	While their $\sq$-balancing schemes are in fact variable-length schemes, Skachek and Immink did not provide an analysis of the average redundancy of their schemes and instead argued that $\log_2 n + O(1)$ redundant bits are sufficient in the worst case (when $\sq$ is constant). 
	%\vspace{2mm}
	
	In the case of $\sq=0$, we also study balanced codes with error-correcting capabilities. 
	As mentioned earlier, we look a balancing technique introduced by Chee~\etal{}~\cite{chee2020} that allows one to construct balanced codes from cyclic codes.
	
	\begin{theorem}[{\hspace*{-0.2mm}\cite{chee2020}}]\label{thm:chee}
		Given a $[n,k,d]$-cyclic code, there exists an $(n,d)$-balanced code of size at least $2^k/n$.
	\end{theorem}

	A key ingredient to the proof of Theorem~\ref{thm:chee} is a {\em cyclic balancing technique} which we described in detail in Section~\ref{sec:ecc}.
	In the same paper~\cite[Section~VI]{chee2020}, the authors also proposed an efficient way of encoding roughly $2^k/n$ messages into balanced codewords. However, to perform such an encoding, the authors require a polynomial with certain algebraic properties. In this work, we propose a variable-length balancing scheme that allows us to remove the requirement.

	\subsection{Our Contributions}
	\begin{enumerate}[(I)]
	\item In this paper, we amalgamate the variable-length scheme in~\cite{immink2010very} with the $\sq$-balancing schemes in~\cite{skachek2014constant} to obtain new variable-length $\sq$-balancing schemes. We formally describe Schemes~A and~B in Sections~\ref{sec:schemeA} and~\ref{sec:schemeB}, respectively.
	\item  Crucially, our objective is to provide a sharp analysis of the average redundancy of our $\sq$-balancing schemes. In Section~\ref{sec:schemeA}, we outline our analysis strategy. Section~\ref{sec:lattice} then provides the connection with lattice path combinatorics, a detailed proof that the fraction of bad words is negligible, and finally, a closed expression~\eqref{eq:rhoA} for an upper bound on the average redundancy. 
	\item Unfortunately, we were unable to give an asymptotic estimate for \eqref{eq:rhoA}. Hence, we make a small modification to obtain Scheme~B and demonstrate Theorem~\ref{thm:main}. In Table~\ref{table:balanced}, we compare the average redundancy of Scheme~B with those in prior work when $\sq=0$.  %\todo{}
	
\begin{theorem}\label{thm:main}
Fix $\sq$ to be constant. Scheme B is an $(n,n-1,\rho_B;\sq)$-variable-length balancing scheme where
\begin{equation}\label{eq:rhoB}
	\rho_B \begin{cases}
		\lesssim \frac12 \log_2 n + 2.526\ldots, & \mbox{if $\sq>0$},\\
		\sim \frac12 \log_2 n + 0.526\ldots, & \mbox{if $\sq=0$}.
	\end{cases}
\end{equation}
Here, we write $f(n)\lesssim g(n)$ if $\lim_{n\to\infty} f(n)/g(n)\le 1$.
\end{theorem}
	
	\item We study error-correcting balanced codes in the special case where $\sq=0$.
	Specifically, we propose a variable-length balancing scheme, Scheme~C, and analyse its average redundancy $\rho_\C$.
	In this case, the value of $\rho_\C$ depends on the underlying code. 
	When $\C=\{0,1\}^{n-1}$, we obtained a closed expression using lattice-path combinatorics.
	On the other hand, for the general case, we provide a method to compute $\rho_\C$ efficiently.
	\end{enumerate}

	\begin{table}[t]
		\centering
		%\footnotesize
		\setlength{\tabcolsep}{3pt} 
		\renewcommand*{\arraystretch}{1.1}
		%\begin{tabular}{|c|p{12mm}|p{12mm}|p{18mm}|p{15mm}|p{18mm}|}
		\begin{tabular}{|c|C{28mm}|C{28mm}|C{28mm}|C{20mm}|C{20mm}|C{28mm}|}
				\hline
			$n$ & Scheme~1 in~\cite{Weber2010} & Scheme~2 in~\cite{Weber2010} & Scheme in~\cite{immink2010very} & Scheme~B 
			& Scheme~C & Optimal Redundancy~\eqref{eq:optimal-exact} \\
			\hline
			
			8   & 2.65    & 2.04    & 1.90    & 2.01    & 2.12  & 1.87\\
			16  & 3.53    & 2.60    & 2.38    & 2.52    & 2.81  & 2.35\\
			32  & 4.44    & 3.15    & 2.87    & 3.02    & 3.59  & 2.84\\
			64  & 5.37    & 3.7     & 3.36    & 3.52    & 4.42  & 3.33\\
			128 & 6.32    & 4.25    & 3.86    & 4.02    & 5.30  & 3.83\\
			256 & 7.29    & 4.78    & 4.36    & 4.53    & 6.22  & 4.33\\
			512 & 8.27    & 5.31    & 4.86    & 5.03    & 7.15  & 4.83\\\hline  
		\end{tabular}
		\caption{Average Redundancy of Variable-Length Balancing Schemes when $\sq=0$. We remark that Scheme~C is equipped with error-correcting capabilities (see Section~\ref{sec:ecc}).}
		\label{table:balanced}
	\end{table}
	
	\begin{table}[t]
		\centering
		%\small
		\setlength{\tabcolsep}{3pt} 
		\renewcommand*{\arraystretch}{1.1}
	%	\begin{tabular}{|c|p{20mm}|p{20mm}|p{20mm}|p{16mm}|}
	\begin{tabular}{|c|C{40mm}|C{35mm}|C{35mm}|C{45mm}|C{28mm}|}
			\hline
			$n$ & Lower Bound on Redundancy for Schemes in~\cite{skachek2014constant} & 
			Scheme~A Eq.~\eqref{eq:schemeA} &
			Scheme~B Eq.~\eqref{eq:schemeB} & Optimal Redundancy Eq.~\eqref{eq:optimal-exact} \\
			\hline
			
			16 &	0.18 &	17.39 &	29.91 &	9.09\\
			32 &	1.80 &	13.70 &	23.16 &	6.06\\
			64 &	4.70 &	8.03 &	12.06 &	4.94\\
			128 &	6.83 &	5.77 &	6.79 &	4.63\\
			256 &	8.00 &	6.19 &	6.62 &	4.73\\
			512 &	9.00 &	6.83 &	7.11 &	5.03\\
			1000 &	9.97 &	7.35 &	7.57 &	5.41\\  
			\hline
		\end{tabular}
		\caption{Average Redundancy of Variable-Length Balancing Schemes when $\sq=6$. For a fair comparison, we provide a lower bound for the schemes in~\cite{skachek2014constant}. Here, the values are $(1-\frac{D(n,\sq)}{2^n})\log n$. In other words, we ignore the number of redundant bits used to encode bad words.}
		\label{table:qbalanced}
	\end{table}

	\section{First Variable-Length $\sq$-Balancing Scheme}
	\label{sec:schemeA}
	
	We formally describe our variable-length schemes in this section.
	We first present Scheme~A and outline a strategy to analyze the average redundancy of Scheme~A. % and prove Theorem~\ref{thm:main} \todo{something}.
	As mentioned earlier, the main obstacle in the direct application of Knuth's technique is the existence of bad words. 
	Next, we define bad and good words and classify good words into two types.
	
	\begin{definition}
	A word $\bx$ is {\em bad} if $T(\bx,\sq)=T(\overline{\bx},\sq)=\varnothing$.
	Otherwise, $\bx$ is {\em good}. 
	Furthermore, $\bx$ is {\em Type-$1$-good} if $T(\bx,\sq)\ne\varnothing$;
	and $\bx$ is {\em Type-$0$-good} if $T(\bx,\sq)=\varnothing$ and $T(\overline{\bx},\sq)\ne\varnothing$.
	\end{definition}
	
	In Section~\ref{sec:badwords}, we show the following lemma that provides a simple sufficient condition for a word to be good.
	
	\begin{lemma}\label{lem:good}
	If $\weight(\bx)\ge n/2 + \sq$ or $\weight(\bx)\le n/2 - \sq$, then $\bx$ is Type-$1$-good. 
%	If $\weight(\bx)\le n/2 - \sq$, then $\bx$ is Type-$0$-good. 
	\end{lemma}

	Here, we use Lemma~\ref{lem:good} to provide a simple way to obtain a good word from a bad one using $2\sq$ redundant bits. 
	We remark that in \cite{skachek2014constant}, more sophisticated techniques are used to handle these bad words with less redundancy. %\todo{Tu, please check this!}
	
	\begin{lemma}\label{lem:bad}
		Suppose that $\bx$ is bad.
		If $\bx'=\bx^{[n-2\sq]}$,  then either $\bx'0^{2\sq}$ or $\bx'1^{2\sq}$ is Type-$1$-good. We say that $\bx$ is Type-$i$-bad if $\bx'i^{2\sq}$ is Type-$1$-good. 
 	\end{lemma}
 
 \begin{proof}
 If $\weight(\bx')\le n/2-\sq$, then appending $0$s does not alter the weight. So, $\bx'0^{2\sq}$ has weight at most $n/2-\sq$ and thus, the latter word is Type-1-good by Lemma~\ref{lem:good}. 
 Otherwise, $\weight(\bx')\ge n/2-\sq+1$ and $\bx'1^{2\sq}$ has weight at least $n/2+\sq+1$. Again, Lemma~\ref{lem:good} implies that the latter is a Type-1-good word.
 \end{proof}

	With this lemma, we are ready to define our first variable-length $\sq$-balancing scheme. Henceforth in this paper, for expository purposes,  we drop the ceiling and floor functions and assume that all logarithmic functions return integer values.
	
	\vspace{2mm}

	\noindent{\bf Scheme A}: An $(n,n,\rho_A;\sq)$-balancing scheme\\
	\hspace*{5mm} {\sc Input}: $\bx\in\{0,1\}^n$\\[1mm]
	\hspace*{1mm} {\sc Output}: $\bc\in\B(n,\sq)$, $\bp \in \{0,1\}^*$\\[-5mm]
	\begin{enumerate}[(I)]
	\item Determine if $\bx$ is Type-$i$-good or Type-$i$-bad.
		\begin{itemize}
		\item If $\bx$ is Type-1-good, set $\hat{\bx}=\bx$.
		\item If $\bx$ is Type-0-good, set $\hat{\bx}=\overline{\bx}$.
		\item If $\bx$ is Type-1-bad, set $\hat{\bx}=\bx^{[n-2q]}1^{2q}$.
		\item If $\bx$ is Type-0-bad, set $\hat{\bx}=\bx^{[n-2q]}0^{2q}$.
		\end{itemize}
	By Lemmas~\ref{lem:good} and~\ref{lem:bad}, we have that $T(\hat{\bx},\sq)\ne \varnothing$.
	\item Determine the $\sq$-balanced word $\bc$.
		\begin{itemize}
		\item $\tau \gets \min T(\hat{\bx},\sq)$
		\item $\bc \gets \flip(\hat{\bx},j)$
		\end{itemize}
	\item Determine the prefix $\bp$.	
		\begin{itemize}
		\item Compute $\Gamma^{(A)}_\sq(\bc)$ using \eqref{eq:Gamma} (see Section~\ref{sec:average})
		\item $\bz \gets$ length-$r$ binary representation of the index of $\tau$ in $\Gamma^{(A)}_\sq(\bc)$. Here, $r=\log|\Gamma^(A)_\sq(\bc)|$.
		\item If $\bx$ is Type-$i$-good, we set $\bp\gets 0i\bz$.
		\item If $\bx$ is Type-$i$-bad,  we set $\bp\gets 1i\bz\bx_{[2q]}$.
		\end{itemize}
	\end{enumerate}
	
	When $\sq = 0$, all words are Type-1-good by Lemma~\ref{lem:good}.
	Hence, we can omit Step (I) and in Step (III), we simply set $\bp$ to be $\bz$.
	In this case, we recover Immink and Weber's variable-length scheme in~\cite{immink2010very}.
	
	Roughly speaking, $\Gamma_\sq^{(A)}(\bc)$ denote the set of possible indices/prefixes that can be received with the codeword $\bc$. {\em A priori} the size of  $\Gamma_\sq^{(A)}(\bc)$ is at most $n$ and so, we need at most $\log_2 n$ bits to represent the indices.
	Nevertheless, Immink and Weber observed is that the size of $\Gamma_0^{(A)}(\bc)$ often much smaller than $n$ and thus, the redundancy incurred by some codewords is much smaller than $\log_2 n$ bits.
	A key contribution in~\cite{immink2010very} is to determine the average redundancy incurred by the scheme.
	In Section~\ref{sec:average}, we extend the analysis to the general case where $\sq>0$.

	\subsection{Instructive Example}\label{sec:example}
	
	Fix $n=8$, $\sq=2$, and hence, our target weight is~six. 
	We consider the message $\bx = 1110 0000$.
	For Step (I), using Lemma~\ref{lem:Pi}, we determine that $\bx$ is Type-0-good and set $\hat{\bx}=0001 1111$.
	
	In Step (II), we find that $\tau=1$ and set $\bc=1001 1111$.
	
	In Step (III), we find the prefix $\bp$. 
	Since $R(\bc) = (0,1,0,-1,0,1,2,3,4)$, it follow from~\eqref{eq:Gamma} that $\Gamma_\sq^{(A)}(\bc) = \{0,1,3,6,7,8\}$. Since the size of $\Gamma_\sq^{(A)}(\bc)$ is six, three bits suffice to represent all balancing indices and the representation of $\tau$ is $001$.
	Thus, we set $\bp=00 001$.
	
	We summarize this example and include encodings of other messages in Figure~\ref{fig:schemeA}.
	
	\subsection{Average Redundancy Analysis}
	
	We now analyze the average redundancy of Scheme A.
	To this end, consider a message $\bx$ and let $\bz(\bx)$ and $\bp(\bx)$ be the resulting index representation and prefix, respectively. 
	Observe from Scheme~A, when the word $\bx$ is good, the prefix $\bp(\bx)$ has length $2+|\bz(x)|$. 
	On the other hand, when the word $\bx$ is bad,  the prefix $\bp(\bx)$  has length $2+|\bz(x)|+2q\le 2+2q+\log_2 n$.
	
	Therefore, if we denote the number of bad words in $\{0,1\}^n$ by $D(n,\sq)$, we have that
	\begin{equation}
	\rho_A\le 2 + \frac{1}{2^n}\left(\sum_{\bx \text{ is good}}|\bz(\bx)|\right)
	+ \frac{D(n,\sq)}{2^n}(2\sq+\log_2 n).\label{eq:schemeA}
	\end{equation}
	
	Hence, in the next section, we use lattice path combinatorics to analyse both $D(n,\sq)$ and $\sum_{\bx \text{ is good}}|\bz(\bx)|$. 
	For the latter quantity, we have the following proposition whose proof is deferred to Section~\ref{sec:badwords}.
	
	\begin{proposition}\label{prop:badwords}
		For fixed $\sq$, we have $D(n,\sq)=o(2^n)$.
	\end{proposition}
	
	Therefore, it remains to study the quantity  $\sum_{\bx \text{ is good}}|\bz(\bx)|$ and we derive a closed formula in Section~\ref{sec:average}.
	Then, together with \eqref{eq:schemeA}, we have the following estimate for the average redundancy of Scheme~A.
	
	\begin{theorem}
	Scheme A is an $(n,n,\rho_A;\sq)$
	variable-length balancing scheme where
	\begin{equation}\label{eq:rhoA}
		\rho_A 
			\lesssim \frac{1}{2^n}\sum_{i=0}^{n+1} i\gamma_\sq(i,n)\log_2 i \,,
	\end{equation}
	\noindent where
	\begin{align*}%\label{eq:gammaq}
		\gamma^{(A)}_\sq(i,n) & =  \sum_{t=0}^{i-1} 
		\Big(G(i-1,t) - G(i-2,t) - G(i-2,t-1) + G(i-3,t-1) \Big)\, ,\text{ and } \\
		G(i,t) & = \frac{2^{n+2}}{i+2} \sum_{k=1}^{\floor{(i+1)/2}} \cos^n{\frac{\pi k}{i+2}}\sin{\frac{\pi k (t+1)}{i+2}} \sin{\frac{\pi k (t+1-2\sq)}{i+2}}\,.
	\end{align*}
	%}
%\noindent where $G(i,t)$ is given by the expression
%{
%\begin{equation}%\label{eq:Git}
%	\frac{2^{n+2}}{i+2} \sum_{k=1}^{\floor{(i+1)/2}} \cos^n{\frac{\pi k}{i+2}}\sin{\frac{\pi k (t+1)}{i+2}} \sin{\frac{\pi k (t+1-2\sq)}{i+2}}\,.
%\end{equation}
	%$\gamma_\sq(i,n)$ is computed using Theorem~\ref{thm:gammaq}.
%\todo{State \eqref{eq:gammaq} and \eqref{eq:Git}  first.}	
	\end{theorem}		
	
	We defer the detailed derivation of $\gamma_\sq(i,n)$ in Section~\ref{sec:lattice} (see Theorem~\ref{thm:gammaq}).
	It follows from the above expression that the average redundancy $\rho_A$ can be computed in polynomial time.
	%Hence, the average redundancy $\rho_A$ can be computed in polynomial time.
	%However, an asymptotically tight analysis of the expressions in \eqref{eq:rhoA}, \eqref{eq:gammaq} and \eqref{eq:Git} remains elusive. 
	However, an asymptotically tight analysis of these expressions remains elusive. So, in Section~\ref{sec:schemeB}, we make a simple modification to Scheme~A and show that the resulting scheme has average redundancy at most $\frac12 \log_2 n +O(1)$.
	
	%\todo{Tu: Reorganize. Now, readers need go back and forth to understand.}

	\section{Analysis using Lattice Path Combinatorics}
	\label{sec:lattice}
	
	In this section, we complete the analysis of $\rho_A$.
	Results in this section are based on a classic combinatorial problem -- lattice path enumeration. We remark that Immink and Weber applied similar methods for their analysis of variable-length $0$-balancing schemes~\cite{immink2010very, Weber2010}. Here, we not only extend the analysis to the case where $\sq>0$, but also use lattice path combinatorics to enumerate bad words.
	
	%Here, we will introduce the simplest enumeration problem for lattice path and its main theorem \cite{krattenthaler2015lattice}. 
	
	\begin{definition}
		A path in the integer lattice plane $\mathbb{Z}^2$ is {\em simple} if it starts from a lattice point and consists of horizontal $\rightarrow$ and vertical $\uparrow$ unit steps in the positive direction.
	\end{definition}

	Since the 1850s, lattice paths have been extensively studied and we refer the reader to Krattenthaler~\cite{krattenthaler2015lattice} for a comprehensive survey of the history, applications, and related results. 
	In this paper, we are interested in counting the following quantities.
	Suppose that $a$, $b$, $c$, $d$, $s$, and $t$ are integers with the following properties: $a\le c$, $b\le d$, $a+s \le b \le a+t$, and $c+s \le d\le c+t$. 
	We then let $\bL\big((a,b) \rightarrow (c,d)\big)$ denote the number of simple paths from $(a,b)$ to $(c,d)$, 
	while $\bL\big((a,b) \rightarrow (c,d);s,t\big)$ denotes the number of simple paths from $(a,b)$ to $(c,d)$ 
	that stay below\footnote{Here, we include simple paths that touch both lines $Y=X+t$ and $Y=X+s$.} 
	the line $Y=X+t$ and above the line	$Y=X+s$. 
	We have the following enumeration results.
	
	\begin{theorem}[\cite{krattenthaler2015lattice}]\label{thm:lattice}
	Set $\Delta_x=c-a$ and $\Delta_y=d-b$.
	\begin{equation}
	\bL\big((a,b) \rightarrow (c,d)\big) = \binom{\Delta_x+\Delta_y}{\Delta_x} \label{eq:lattice}
	\end{equation}
	Suppose further that $\Delta_h=t-s+2$.
	\begin{align}
	%&\bL\Big((a,b)\ra(c,d);s,t\Big) \notag \\
	\bL\Big((a,b)\ra(c,d);s,t\Big)
	&= \sum_{k \in \mathbb{Z}}\binom{\Delta_x+\Delta_y}{\Delta_x-k\Delta_h} 
	- \binom{\Delta_x+\Delta_y} {c-b-k\Delta_h+t+1} \label{eq:lattice-bdd1}\\
	& = \sum_{k=1}^{(\floor{(\Delta_h-1)/2}} \frac{4}{\Delta_h}  \left(2\cos{\frac{\pi k}{\Delta_h}}\right)^{\Delta_x+\Delta_y}%\notag\\& \hspace{5mm}\times
	\sin{\frac{\pi k (a-b+t+1)}{\Delta_h}} \sin{\frac{\pi k (c-d+t+1)}{\Delta_h}} \label{eq:lattice-bdd2}
	\end{align}
	Suppose that $t=\infty$ and $s=0$. That is, we count all paths above $Y=X$. 
	So, $a\le b$ and $c\le d$.
	\begin{equation}
	\bL\Big((a,b)\ra(c,d);0,\infty\Big) =\binom{\Delta_x+\Delta_y}{\Delta_x} - \binom{\Delta_x+\Delta_y}{d-a+1}\,. \label{eq:lattice-bddabove}
	\end{equation}
	\end{theorem}

	Next, we describe a transformation $\Pi$ that maps binary words $\{0,1\}^n$ of a certain weight to length-$n$ simple paths between certain lattice points in $\mathbb{Z}^2$. Specifically, for any binary word $\bx=x_1x_2\ldots x_n$, we define a simple path $\Pi(\bx)$ where
	\begin{itemize}
		\item $\Pi(\bx)$ starts from $(0,\weight(\bx)-m)$ (recall that $n=2m$);
		\item When $x_i=0$, we move a vertical unit $\uparrow$;
		\item When $x_i=1$, we move a horizontal unit $\rightarrow$.
	\end{itemize}
	
	In Figure~\ref{fig:schemeA}, we provide some examples of the transformation $\Pi$.
	The following lemma summarizes certain properties of $\Pi$.
	\begin{lemma}\label{lem:Pi}\hfill
	\begin{enumerate}[(i)]
	\item Fix $\sq$ and recall that  $\B(n,\sq)$ is the set of words with weight $m+\sq$. 
	Let ${\cal P}$ be the set of all simple paths from $(0,\sq)$ to $(m+\sq,m)$.
	Then $\Pi$ is a bijection from $\B(n,\sq)$ to $\cal P$.
	\item Let $\bx=x_1x_2\ldots x_n$. Suppose that the lattice points on $\Pi(\bx)$ are $(X_0,Y_0)$, $(X_1,Y_1)$, \ldots, $(X_n,Y_n)$. Then for $j\in\bbracket{n}$,  the imbalance of $\flip(\bx,j)$ is given by $2(Y_j-X_j)$. 
	\end{enumerate}
	\end{lemma}

	\begin{proof}
	\begin{enumerate}[(i)]
	\item This is straightforward from the definition of $\Pi$.
	\item We prove this by this induction. When $j=0$, the imbalance of $\bx$ is $2\weight(\bx)-n = 2(\weight(\bx)-m) = 2(Y_0-X_0)$.
	
	Next, we assume that the imbalance of $\flip(\bx,j-1)$ is given by $2(Y_{j-1}-X_{j-1})$.
	If $x_j=0$, then $Y_j$ increases by one, and the imbalance of $\flip(\bx,j)$ increases by two. 
	Similarly, if $x_j=1$, then $X_j$ increases by one, and the imbalance of $\flip(\bx,j)$ decreases by two.
	In both cases,  imbalance of $\flip(\bx,j-1)$ is given by $2(Y_{j-1}-X_{j-1})$ and the proof is complete. \qedhere
	\end{enumerate}
	\end{proof}
	%\begin{comment}
\begin{figure*}
\centering
\subfloat[Lattice paths and their corresponding binary~words]{
	\small
\begin{tikzpicture}[scale=0.50]
	\draw[very thin,color=gray,step=1, dashed] (0,-3) grid (7,4);
	\draw[thick] (0,0) -- (7,0);
	\draw[thick] (0,-3) -- (0,4);
	\draw[thin, green] (0,2) --  (2,4);
	\draw[thin, green] (0,-2) -- (6,4);
	\node at (-0.4, -3) {-3};
	\node at (-0.4, -2) {-2};
	\node at (-0.4, -1) {-1};
	\node at (-0.4, 0) {0};
	\node at (-0.4, 1) {1};
	\node at (-0.4, 2) {2};
	\node at (-0.4, 3) {3};
	\node at (0, 4.3) {$Y=X+2$};
	\node at (7, 4.3) {$Y=X-2$};
	\draw[line width=0.9mm, color = blue, ->](0,-2) -- (0,-1)--(2,-1)--(2,4);
	\draw[line width=0.5mm, color = orange, ->] (0,-1) -- (3,-1) -- (3,4);
	\draw[line width=0.5mm, color = red, ->](0,0) -- (0,1)--(1,1)--(2,1)--(2,2)--(2,3)--(3,3)--(4,3)--(4,4);
\end{tikzpicture}}~
\subfloat[Encoding of messages using Scheme~A]{

\renewcommand{\arraystretch}{1.5}
\small
\begin{tabular}[b]{ccc cccl}
\hline
$\bx$ & Type & $\hat{\bx}$ & $\tau$ & $\bc$ & $\Gamma_\sq(\bc)$ & $\bp$ \\
\hline
{\color{blue}$0110 0000$} &
Type-$1$-good & 
$0110 0000$ &
$8$ & 
$1001 1111$ &
$\{0,1,3,6,7,8\}$&
 $01\,101$ \\ 
 
 {\color{orange}$1110 0000$} &
 Type-$0$-good & 
 $0001 1111$ &
 $1$ & 
 $1001 1111$ &
 $\{0,1,3,6,7,8\}$&
 $00\,001$ \\ 
 
 {\color{red}$0110 \underline{0110}$} &
 Type-$0$-bad & 
 $0110 0000$ &
 $8$ & 
 $1001 1111$ &
 $\{0,1,3,6,7,8\}$&
 $10\,101\, \underline{0110}$ \\ 
  \hline\\ 
\end{tabular}
}
\vspace{-3mm}
\caption{Instructive example of Scheme A.\vspace{-6mm}
}
\label{fig:schemeA}
\end{figure*}

	\subsection{Enumeration of Bad Words}\label{sec:badwords}
	First, we apply Lemma~\ref{lem:Pi} to characterize good and bad words.
	
	\begin{lemma}\label{lem:badpaths}
	Fix $\sq>0$ and consider the lines $Y = X + \sq$ and $Y = X -\sq$. 
	Let $\bx\in\{0,1\}^n$ and consider the path $\Pi(\bx)$.
	\begin{enumerate}[(i)]
	\item $\bx$ is bad if and only if $\Pi(\bx)$ is strictly between $Y = X + \sq$ and $Y = X - \sq$. 
	\item $\bx$ is Type-1-good if and only if $\Pi(\bx)$ touches the line $Y = X + \sq$. 
	\item $\bx$ is Type-0-good if and only if $\Pi(\bx)$ does not touch the line $Y = X + \sq$, but touches the line $Y=X-\sq$. 
	\end{enumerate}
	\end{lemma}
	
	Therefore, the number of bad words is immediate from Lemma~\ref{lem:badpaths}(i) and Theorem~\ref{thm:lattice}.
	
	\begin{proposition}\label{prop:badwords-exact}
	We have that 
	%{\footnotesize 
	\begin{equation}
	D(n,\sq) = \sum_{y=1-\sq}^{\sq-1} \sum_{k=1}^{\sq -1} \frac{4}{2\sq}(2\cos{\frac{\pi k}{2\sq}})^{2m}\sin{\frac{\pi k (\sq - y)}{2\sq}}\sin{\frac{\pi k (\sq + y)}{2\sq}}
	\end{equation}
	%}
	\end{proposition} 
	
	\begin{proof}
	Let $\bx$ be a word with weight $m+y$.
	Then by Lemma~\ref{lem:Pi}, we have that $\Pi(\bx)$ starts from $(0,y)$ and ends at $(m+y,m)$. Then it is straightforward to check that when $y\ge \sq$, the path $\Pi(\bx)$ touches the line $Y=X+\sq$. Similarly, when $y\le -\sq$, the path $\Pi(\bx)$ necessarily touches the line $Y=X+\sq$. Therefore, whenever $y\le -\sq$ or $y\ge \sq$, $\bx$ is necessarily Type-1-good. This also proves Lemma~\ref{lem:good}.
	
	Hence, we consider the case where $1-\sq \le  y\le \sq-1$. Then the bad words with weight $m+y$ correspond to the paths that start from $(0,y)$ and end at $(m+y,m)$ and lie strictly between $Y = X + \sq$ and $Y = X - \sq$. Then \eqref{eq:lattice-bdd2}  states that this number is given by $\sum_{k=1}^{\sq -1} \frac{4}{2\sq}(2\cos{\frac{\pi k}{2\sq}})^{2m}\sin{\frac{\pi k (\sq - y)}{2\sq}}\sin{\frac{\pi k (\sq + y)}{2\sq}}$. Summing over all valid values of $y$ recovers the proposition.
	\end{proof}

	Next, we estimate asymptotically the number of bad words. Note that each summand is at most $\frac{4}{2\sq}(2\cos{\frac{\pi k}{2\sq}})^{n}$. Now, since $0<k<\sq$, we have that $0< \frac{\pi k}{2\sq} < \pi/2$. Then the monotone decreasing property of the cosine in this region shows that each summand is at most $\frac{4}{2\sq}(2\cos{\frac{\pi}{2\sq}})^n$.
	Since there are at most $\frac{4(\sq-1)^2}{\sq}$ summands, we have $D(n,\sq)\le \frac{4(\sq-1)^2}{\sq} (2\cos{\frac{\pi}{2\sq}})^n=o(2^n)$. This completes the proof of Proposition~\ref{prop:badwords}.

	\subsection{Average Redundancy of Good Words}\label{sec:average}
	
	In this section, we provide a closed formula for the quantity $\sum_{\bx\text{ is good}} |\bz(\bx)|$ defined in \eqref{eq:schemeA}.
	Immink and Weber first investigated this quantity in their proposed variable-length $0$-balancing scheme~\cite{immink2010very}. In their paper, instead of studying the set of balancing indices associated with a \textit{message} $\bx$, they investigated the set of $0$-balancing indices that can be received with a \textit{codeword} $\bc$. Formally, for $\bc \in \B(n,\sq)$, we consider the following set:
	\begin{equation*}
		\Gamma_\sq^{(A)}(\bc) \triangleq \{j \in \bbracket{n}: \flip(\bx,j)=\bc\, ,\, \tau_\sq(\bx)=j \text{ for } \bx \in \{0,1\}^n\}\,.
	\end{equation*}
	
	Here, $\tau_\sq(\bx)$ is defined to be $\min T(\bx,\sq)$. 
	It turns out that given $\bc \in \B(n,\sq)$, we can determine $\Gamma_\sq^{(A)}(\bc)$ efficiently.
	Specifically, when $\sq=0$, Immink and Weber provided a simple characterization of $\Gamma_\sq^{(A)}(\bc)$ using the notion of \textit{running sum}.
	
	\begin{definition}
		Let $\bx=x_1x_2\ldots x_n \in \{0,1\}^n$. The \textit{running sum} of $\bx$, denoted by $R(\bx)$, is a length-$(n+1)$ integer-valued vector indexed by $\bbracket{n}$ defined by 
		\begin{equation*}
			R(\bx)_i = 
			\begin{cases}
				0, &\text{if } i = 0, \\
				R(\bx)_{i-1}+ (-1)^{x_i+1}, &\text{if } i > 0.
			\end{cases}
		\end{equation*}
	\end{definition}
	
	The following proposition generalizes a result in~\cite{immink2010very} for the case where $\sq\ge 0$.
	
	\begin{proposition}\label{prop:Gamma}
		Let $\bc \in \B(n,\sq)$. Then
		\begin{equation}
			\Gamma_\sq^{(A)}(\bc) = \{i \in \bbracket{n}: R(\bc)_i \neq R(\bc)_j \text{ for all } j < i\}\,.\label{eq:Gamma}
		\end{equation}
	\end{proposition}
	\begin{proof}
		Let $\bc = c_1c_2\cdots c_n$. Then $\sum_{k=1}^n c_k = m+\sq$.
		
		Now, let $\bx = \flip(\bc,i)$ and suppose that $i \notin \Gamma_\sq(\bc)$. This means that there is an index $j < i$ where $\flip(\bx,j) \in \bB(n,\sq)$. 
		
		Since $\flip(\bx,j)=c_1\cdots c_j\overline{c_{j+1}\cdots c_{i}}c_{i+1}\cdots c_n$, we have that 
		$\sum_{k=1}^n c_k + (i-j)- 2\sum_{k=j+1}^i c_k = m+\sq$. 		
		Hence, 
		\[\left(2\sum_{k=j+1}^i c_k\right) - (i-j)=\sum_{k=j+1}^i (2c_k-1) =\sum_{k=j+1}^i (-1)^{c_k+1} = 0\,,\] which is equivalent to $R(\bc)_i = R(\bc)_j$.
	\end{proof}

	Next, for $1\le i\le n+1$, we consider the following subset of words in $\B(n,\sq)$:
	\begin{equation}
		\cE^{(A)}_\sq(i,n) = \{\bc \in \B(n,\sq): |\Gamma_\sq^{(A)}(\bc)|=i\}
	\end{equation}
	and determine its size $\gamma^{(A)}_\sq(i,n)\triangleq |\cE^{(A)}_\sq(i,n)|$.
	
	To characterize the words in $\cE_\sq(i,n)$ as lattice paths, we introduce the notion of {\em  width}.
	
	\begin{definition}\label{def:width}
	Consider a lattice path $\boldsymbol{\pi}$.
	Suppose that $t_{\min}$ and $s_{\max}$ are the smallest and largest integers such that $\boldsymbol{\pi}$ lies in between $Y=X+t_{\min}$ and $Y=X+s_{\max}$. Then the {\em width} of $\boldsymbol{\pi}$ is defined to be $t_{\min}-s_{\max}$.
	\end{definition}
	
	%We now characterize words in $\cE_\sq(i,n)$.
	
	\begin{lemma}\label{lem:pathwidth}
	Fix $\sq$ and $1\le i \le n+1$. 
	Then $\bx\in\cE^{(A)}_\sq(i,n)$ if and only if the path $\Pi(\bx)$ has width exactly $i-1$.
	\end{lemma}
	
	Hence, determining $\gamma_\sq^{(A)}(i,n)$ is equivalent to enumerating lattice paths with a certain width.
	Therefore, using principles of inclusion and exclusion with Theorem~\ref{thm:lattice}, we have the following result.
	Here, we reproduce the expressions in Theorem~\ref{thm:gammaq} for convenience.
	
	\begin{theorem}\label{thm:gammaq}
	We have that %\todo{recall $\gamma^{(A)}_\sq$}
	%{ 
	\begin{equation}\label{eq:gammaq}
		\gamma^{(A)}_\sq(i,n) =  \sum_{t=0}^{i-1} 
		\Big(G(i-1,t) - G(i-2,t) - G(i-2,t-1) + G(i-3,t-1) \Big)\, , 
	\end{equation}
	%}
	\noindent where $G(i,t)$ is given by the expression
	%{
	\begin{equation}\label{eq:Git}
	 \frac{2^{n+2}}{i+2} \sum_{k=1}^{\floor{(i+1)/2}} \cos^n{\frac{\pi k}{i+2}}\sin{\frac{\pi k (t+1)}{i+2}} \sin{\frac{\pi k (t+1-2\sq)}{i+2}}\,.
	\end{equation}
    %}
	\end{theorem}
	
	\begin{proof}
	Let $\bc \in \cE^{(A)}_\sq(i,n)$. Then using the transformation $\Pi$, we have that $\Pi(\bc)$ is a simple path from $(0,\sq)$ to $(m+\sq,m)$. Then Lemma~\ref{lem:pathwidth} states that  $\Pi(\bc)$ has width exactly $i-1$.
	In other words, if we set $t=t_{\min}$ as in Definition~\ref{def:width}, then $s_{\max}=t-(i-1)$.	
	Hence, for this fixed value of $t$, we are interested in the number of paths that touch {\em both} $Y=X+t$ and $Y=X+t-(i-1)$. % for a fixed $t$.
	
	Now, applying Theorem~\ref{thm:lattice} with \eqref{eq:Git}, we observe that 
	$G(i-1,t)= \bL\Big((0,0)\ra(m+\sq,m);t-(i-1),t\Big)$ counts all paths that are between $Y=X+t$ and $Y=X+t-(i-1)$.
	Next, we apply the principle of inclusion and exclusion.
	Specifically, we exclude the paths between $Y=X+t$ and $Y=X+t-(i-2)$, between $Y=X+(t-1)$ and $Y=X+t-(i-1)$, and include the paths between $Y=X+(t-2)$ and $Y=X+t-(i-2)$. These quantities are given by $G(i-2,t)$, $G(i-2,t-1)$ and $G(i-3,t-1)$, respectively.
	
	Summing over all possible values of $t$, we obtain \eqref{eq:gammaq}.
	\end{proof}

	Note that when $\sq=0$, we recover \eqref{eq:gamma}.
	Inserting the expression of $\gamma_\sq^{(A)}(i,n)$ into \eqref{eq:rhoA}, we obtain a closed formula for an upper bound of $\rho_A$. 
	As \eqref{eq:Git} appears unamenable to asymptotic analysis, we introduce our second variable-length balancing scheme.
		
	\section{Second Variable-Length $\sq$-Balancing Scheme}
	\label{sec:schemeB}
	
	In this section, we define Scheme~B and analyze its average redundancy with methods similar to previous sections. We then complete the asymptotic analysis and prove Theorem~\ref{thm:main}. As mentioned earlier, Scheme~B is similar to Scheme~A and the main difference is that we consider a message $\bx$ of length $n-1$. Hence, we first encode into codeword $\bc'$ with weight $m+\sq -1$ or $m+\sq$. Then, we append one extra redundant bit so that the resulting codeword $\bc$ has length $n$ and weight $m+\sq$. In other words, $\bc\in\B(n,\sq)$. Surprisingly, this modification simplifies the analysis of average redundancy and allows us to prove that the average redundancy of Scheme~B is within an additive constant of $0.2$ from the optimal value in~\eqref{eq:optimal-red} (when $\sq=0$).
	
	As before, for a word $\bx\in\{0,1\}^{n-1}$, we define the set of indices $T_B(\bx,\sq)$ and use it to classify $\bx$ as good or bad.
	
	\begin{definition}
		Let $\bx\in\{0,1\}^{n-1}$. Then we set
		{ 
		\begin{equation} \label{eq:TB}
		T_B(\bx,\sq) \triangleq \{j \in \bbracket{n-1}: \weight(\flip(\bx,j)) \in\{m+\sq-1,m+\sq\}\}.
		\end{equation}
		}Then $\bx$ is {\em bad} if $T_B(\bx,\sq)=T_B(\overline{\bx},\sq)=\varnothing$.
		Otherwise, $\bx$ is {\em Type-$1$-good} if $T_B(\bx,\sq)\ne\varnothing$;
		and $\bx$ is {\em Type-$0$-good} if $T_B(\bx,\sq)=\varnothing$ and $T_B(\overline{\bx},\sq)\ne\varnothing$.
	\end{definition}
	
	Using similar methods, we obtain the analogue of Lemma~\ref{lem:bad}.
	
	\begin{lemma}\label{lem:badB}
		Suppose that $\bx$ is bad.
		If $\bx'=\bx^{[n-2\sq-3]}$,  then either $\bx'0^{2\sq-2}$ or $\bx'1^{2\sq-2}$ is Type-$1$-good. We say that $\bx$ is Type-$i$-bad if $\bx'i^{2\sq-1}$ is Type-$1$-good. 
	\end{lemma}
	
	We now formally describe Scheme~B.	
	\vspace{2mm}

	\noindent{\bf Scheme B}: An $(n,n-1,\rho_B;\sq)$-balancing scheme\\
	\hspace*{5mm} {\sc Input}: $\bx\in\{0,1\}^{n-1}$\\[1mm]
	\hspace*{1mm} {\sc Output}: $\bc'\in\B(n,\sq)$, $\bp \in \{0,1\}^*$\\[-5mm]
	\begin{enumerate}[(I)]
		\item Determine if $\bx$ is Type-$i$-good or Type-$i$-bad.
		\begin{itemize}
			\item If $\bx$ is Type-1-good, set $\hat{\bx}=\bx$.
			\item If $\bx$ is Type-0-good, set $\hat{\bx}=\overline{\bx}$.
			\item If $\bx$ is Type-1-bad, set $\hat{\bx}=\bx^{[n-2q-3]}1^{2q-2}$.
			\item If $\bx$ is Type-0-bad, set $\hat{\bx}=\bx^{[n-2q-3]}0^{2q-2}$.
		\end{itemize}
		Hence, we have that $T_B(\hat{\bx},\sq)\ne \varnothing$.
		\item Determine the $\sq$-balanced word $\bc'$.
		\begin{itemize}
			\item $\tau \gets \min T_B(\hat{\bx},\sq)$
			\item $\bc \gets \flip(\hat{\bx},j)$
			\item If $\weight(\bc)=m+\sq-1$, we append 1 to $\bc$ to obtain $\bc'$.
			\item If $\weight(\bc)=m+\sq$, we append 0 to $\bc$ to obtain $\bc'$.
		\end{itemize}
		\item Determine the prefix $\bp$.	
		\begin{itemize}
			\item Compute $\Gamma_\sq^{(B)}(\bc)$ using Proposition~\ref{prop:GammaB}
			\item $\bz \gets$ length-$r$ binary representation of the index of $\tau$ in $\Gamma^{(B)}_\sq(\bc)$. Here, $r=\log|\Gamma^{(B)}_\sq(\bc)|$.
			\item If $\bx$ is Type-$i$-good, we set $\bp\gets 0i\bz$.
			\item If $\bx$ is Type-$i$-bad,  we set $\bp\gets 1i\bz\bx_{[2q-2]}$.
		\end{itemize}
	\end{enumerate}

	As before, to estimate the average redundancy $\rho_B$, we determine the number of bad words and the average redundancy of good words. Specifically, let $\bx$ be a message and set $\bz(\bx)$ and $\bp(\bx)$ to be the resulting index representation and prefix, respectively. 
	If we denote the number of bad words in $\{0,1\}^{n-1}$ by $D_B(n-1,\sq)$, we have that
	%\vspace{-2mm}
	%{ 
	\begin{equation}
		\rho_B\le 3 + \frac{1}{2^{n-1}}\left(\sum_{\bx \text{ is good}}|\bz(\bx)|\right)
		+ \frac{D_B(n-1,\sq)}{2^{n-1}}(2\sq-2+\log_2 n).\label{eq:schemeB}
	\end{equation}
	%}

	To determine $D_B(n-1,\sq)$, we use the following characterization.
	\begin{lemma}\label{lem:badpathsB}
	Fix $\sq>0$ and let $\bx\in\{0,1\}^{n-1}$.
	Then  $\bx$ is bad if and only if $\Pi(\bx)$ is strictly between $Y = X + \sq-1$ and $Y = X - \sq+1$. 
	\end{lemma}

	Then we have the following corollary.
	
	\begin{corollary}\label{cor:badwordsB}
	For fixed $\sq$, $2D_B(n-1,\sq)\le D(n,\sq)=o(2^n)$.
	\end{corollary}

	\begin{proof}
	Let $\bx$ be a bad word of length $n-1$. Then $\Pi(\bx)$ is a lattice path bounded strictly between $Y = X + \sq-1$ and $Y = X - \sq+1$. Now, append either 0 or 1 to $\bx$ to obtain $\bx'$. Clearly, $\Pi(\bx')$ is now bounded strictly between $Y = X + \sq$ and $Y = X - \sq$ and so, $\bx'$ is bad. Hence, 	$2D_B(n-1,\sq)\le D(n,\sq)$ as desired.
	\end{proof}

	Next, we estimate $\sum_{\bx \text{ is good}}|\bz(\bx)|$. 
	To do so, we set $\tau_\sq(\bx)\triangleq \min {T_B(\bx,\sq)}$ and
%	{\small 
	\begin{equation*}
		\Gamma_\sq^{(B)}(\bc) \triangleq \{j \in \bbracket{n-1}: \flip(\bx,j)=\bc\, ,\, \tau_\sq(\bx)=j \text{ for } \bx \in \{0,1\}^{n-1}\}\,.
	\end{equation*}
%	}
As before, we characterize $\Gamma_\sq^{(B)}(\bc)$ using its running sum.
	
	\begin{proposition}\label{prop:GammaB}
		Suppose that $\weight(\bc)\in\{m+\sq-1,m+\sq\}$.
%		\begin{equation*}
%			\Gamma_\sq(\bc) = 
%			\begin{cases}
%				\{i \in \bbracket{n-1}: R(\bc)_i \geq 0,~R(\bc)_i \neq R(\bc)_j \text{ for } j < i\}, \\ \hspace{40mm}\text{ if } \weight(\bc) = n/2+\sq -1,\\
%				\{i \in \bbracket{n-1}: R(\bc)_i \leq 0,~R(\bc)_i \neq R(\bc)_j \text{ for } j < i\}, \\ \hspace{40mm}\text{ if } \weight(\bc) = n/2+\sq .
%			\end{cases}
%		\end{equation*}
	\begin{equation*}
		\Gamma^{(B)}_\sq(\bc) = 
		\begin{cases}
			\{i \in \bbracket{n-1}: R(\bc)_i \geq 0,~R(\bc)_i \neq R(\bc)_j \text{ for } j < i\}, &\text{ if } \weight(\bc) = n/2+\sq -1,\\
			\{i \in \bbracket{n-1}: R(\bc)_i \leq 0,~R(\bc)_i \neq R(\bc)_j \text{ for } j < i\}, &\text{ if } \weight(\bc) = n/2+\sq .
		\end{cases}
	\end{equation*}
	\end{proposition}
	
	\begin{proof}
	We consider the case where $\weight(\bc)=m+\sq -1$. The proof for $\weight(\bc)=m+\sq$ is similar.
	Proceeding as in the proof of Proposition~\ref{prop:Gamma}, we have that $i\notin \Gamma_\sq^{(B)}(\bc)$ if there exists $j<i$ such that $R(\bc)_i=R(\bc)_j$. 
	
	Next, we claim that if $R(\bc)_i<0$, then $i\notin \Gamma^{(B)}_\sq(\bc)$.
	Now, let $\bx = \flip(\bc,i)$. Since $R(\bc)_i<0$, we have that $c_i=0$ and so, $x_i=1$. Then $\flip(\bx,i-1)$ has weight $m+\sq$ and so, the balancing index for $\bx$ is at most $i-1$. Therefore, $i\notin \Gamma_\sq^{(B)}(\bc)$.
	\end{proof}

	As before, for $1\le i\le n$, we consider the following set of length-$(n-1)$ words:
	\begin{equation}
		\cE^{(B)}_\sq(i,n-1) = \{\weight(\bc) \in \{m+\sq-1,m+\sq\}: |\Gamma^{(B)}_\sq(\bc)|=i\}
	\end{equation}
	and determine its size $\gamma_\sq^{(B)}(i,n-1)\triangleq |\cE^{(B)}_\sq(i,n-1)|$.
	
	\begin{lemma}\label{lem:pathwidthB}
	Fix $\sq$ and $1\le i \le n$. 
	Then $\bx\in\cE^{(B)}_\sq(i,n-1)$ if and only if 
	\begin{enumerate}[(i)]
	\item the path $\Pi(\bx)$ starts from $(0,\sq-1)$ and is above and touching $Y=X+\sq-i$, or,
	\item the path $\Pi(\bx)$ starts from $(0,\sq)$ and is below and touching $Y=X+\sq+i-1$.
	\end{enumerate}
	\end{lemma}
	
	Then using principles of inclusion and exclusion with Theorem~\ref{thm:lattice}, we have a surprisingly clean expression for $\gamma_\sq(i,n-1)$.
	
	\begin{theorem}\label{thm:GammaB}
	Fix $\sq>0$. Then % and $\sq\le i\le n$. Then
	%{ 
	\begin{align*}
		\gamma_\sq^{(B)}(i,n-1)%\\
		= 
		\begin{cases}
			 \frac{2i-2\sq}{n}  \binom{n}{m+i-q},
			 & \text{ if }1\le i\le 2\sq-1,\\
			 \frac{2i-2\sq}{n}  \binom{n}{m+i-q}+\frac{2i+2\sq}{n}  \binom{n}{m+i+q},
			 & \text{ if }2\sq\le i\le m-\sq,\\
			 \frac{2i+2\sq}{n}  \binom{n}{m+i+q},
			 & \text{ if }m-\sq+1\le i\le m+\sq.\\ 
		\end{cases} %\label{eq:gammaqB}
	\end{align*}
	%}
	When $\sq=0$, we have 
	\[\gamma_\sq^{(B)}(i,n-1)=\frac{4i}{n}\binom{n}{m+i}\,.\]
	\end{theorem}

	\begin{proof}
	Let $\bc \in \cE_\sq(i,n-1)$. First, suppose that $\bc$ has weight $m+\sq-1$. 
	Then Lemma~\ref{lem:pathwidthB} states that $\Pi(\bc)$ is above and touching the line $Y=X+\sq-i$. 
	By shifting all the paths down by $\sq-i$ units, we obtain the following equivalence.
	The number of such paths is the same as the number of simple paths that start from $(0,i-1)$, end at $(m+\sq-1,m-\sq+i+1)$ and are above and touching the line $Y=X$.
	
	Proceeding as in the proof of Theorem~\ref{thm:gammaq}, we apply the principle of inclusion and exclusion with \eqref{eq:lattice-bdd2}, and compute the number of such paths to be $\left(\binom{2m-1}{m-\sq}-\binom{2m-1}{m+i-\sq}\right)- \left(\binom{2m-1}{m-\sq}-\binom{2m-1}{m+i-1-\sq}\right)$, which is
	{\small \begin{equation*}
	\binom{2m-1}{m+i-1-\sq}-\binom{2m-1}{m+i-\sq}\\
	= \frac{2i-2\sq}{2m}  \binom{2m}{m+i-\sq}\,.
	\end{equation*}
	}Since we require both $(0,i-1)$ and $(m+\sq-1,m-\sq+i+1)$ to be above the line $Y=X$, we have that $\max(1,2\sq)\le i\le m+\sq$.

	Now, when $\bc$ has weight $m+\sq-1$, Lemma~\ref{lem:pathwidthB} states that $\Pi(\bc)$ is below and touching the line $Y=X+\sq+i-1$. Proceeding as before, the number of such paths is also $\frac{2i+2\sq}{2m}  \binom{2m}{m+i+\sq}$.
	Therefore, by summing these two quantities, we obtain the result.
	\end{proof}
	
	Therefore, we have that
	{
	\begin{align}
		\sum_{\bx \text{ is good}}|\bz(\bx)| & = \sum_{i=\max(1,2\sq)}^{m+q}  \frac{2i-2\sq}{2m}\binom{2m}{m+i-q} (i \log{i}) %\notag\\
		%& \hspace{2mm}+  \sum_{i=\max(1,-2\sq)}^{m-q}  \frac{2i+2\sq}{2m}\binom{2m}{m+i+q}  (i \log{i})\,.\label{eq:goodwordsB}
		 \sum_{i=\max(1,-2\sq)}^{m-q}  \frac{2i+2\sq}{2m}\binom{2m}{m+i+q}  (i \log{i})\,.\label{eq:goodwordsB}
	\end{align}	
	}In what follows, we make use of the following result to obtain an asymptotically sharp estimate of \eqref{eq:goodwordsB}.
	
	\begin{theorem}[{\cite[Theorem 4.9]{sedgewick2013introduction}}]\label{thm:catalan}
	Let $F(x)$ be a polynomially bounded function, that is, $F(x)=O(x^d)$ for some integer $d>0$. Then
\begin{equation*}
\sum_{k=1}^{m} F(k)\frac{\binom{2m}{m-k}}{\binom{2m}{m}}\sim 
\int_{x=0}^\infty e^{-x^2/m} F(x) dx\,.
\end{equation*}
	\end{theorem}

	We present the main result of this section.
	
	\begin{theorem}\label{thm:goodwordsB}
	Fix $\sq$. Then
	\begin{equation}\label{eq:beta}
		\frac{1}{2^{n-1}}\sum_{\bx \text{ is good}}|\bz(\bx)| \lesssim \frac12 \log n + \beta \,,
	\end{equation}
	where $\beta = \frac{2-\ln 4 -\gamma}{\ln 4}-\frac 12 \approx -0.474\ldots$ and $\gamma$ is the Euler-Mascheroni constant.
	\end{theorem}

	\begin{proof}
	Let $Z_1$ be the first summand of \eqref{eq:goodwordsB} and we estimate this quantity.
	First, we set $k=i-\sq$ on the right-hand side  and extend the range of $k$ to obtain
	\begin{align}
	%	&\frac{1}{2^{n-1}}\sum_{\bx \text{ is good}}|\bz(\bx)| \notag \\
	Z_1	&\le  \frac{1}{2^{2m-1}} \sum_{k=0}^{m} \frac{2k(k+\sq)\log (k+\sq)}{2m}\binom{2m}{m+k} \notag \\ 
		& = \frac{2}{m}\frac{\binom{2m}{m}}{2^{2m} } \sum_{k=1}^{m} k(k+\sq)\log (k+\sq)\frac{\binom{2m}{m+k}}{\binom{2m}{m}}\notag \\
		& = \frac{2}{m}\frac{\binom{2m}{m}}{2^{2m} } \sum_{k=1}^{m} (k+\sq)^2\log (k+\sq)\frac{\binom{2m}{m+k}}{\binom{2m}{m}} 
		\,\label{rhs}.
	\end{align}
	
	Next, we have Stirling's estimate of the central binomial~\cite{sedgewick2013introduction}:  
	\begin{equation*}
	\frac{\binom{2m}{m}}{2^{2m}}\sim \frac{1}{\sqrt{\pi m}}\, .
	\end{equation*}

	Then setting $F(x)=(x+\sq)^2\log (x+\sq)$ in Theorem~\ref{thm:catalan}, we have that 
	\begin{align*}
	\sum_{k=1}^{m} (k+\sq)^2\log (k+2+\sq)\frac{\binom{2m}{m+k}}{\binom{2m}{m}}
	\sim 
	\int_{x=0}^\infty e^{-x^2/m} (x+\sq)^2\log (x+\sq) dx.
	\end{align*}
	
	We do a change in the variable and set $z = x/\sqrt{m}$ to obtain the expression 
	$\int_{z=0}^\infty e^{-z^2} (\sqrt{m}z+\sq)^2\log (\sqrt{m}z+\sq) (\sqrt{m} dz)$.
	Since $(\sqrt{m}z+\sq)^2\sim mz^2$ and $\log (\sqrt{m}z+\sq)\sim \log(\sqrt{m}z)$, we have  
	\[Z_1\lesssim m^{3/2} \int_{z=0}^\infty e^{-z^2} z^2 \log z dz + \frac{m^{3/2}\log m}{2} \int_{z=0}^\infty e^{-z^2} z^2 dz\,. \]
	Similarly, for the second summand in~\eqref{eq:goodwordsB}, we obtain the exact same expression.

	Now, we have $\int_{z=0}^\infty e^{-z^2} z^2 \log z dz = \frac{\sqrt{\pi}(2-\ln 4 -\gamma)}{4 \ln 4}$ and 
	$\int_{z=0}^\infty e^{-z^2} z^2 dz=\frac{\sqrt{\pi}}{4}$.
	So, combining everything, we have that the right-hand side of \eqref{rhs} tends to
	\begin{align*}
		\left(\frac{2}{m}\right)\left(\frac{1}{\sqrt{\pi m}}\right)m^{3/2}\left(\frac{\sqrt{\pi}(2-\ln 4 -\gamma)}{4\ln 4} + \frac{\sqrt{\pi}}{4}\left(\frac12\log m\right)\right)
		&= \Bigg(\frac12\log m +\frac{2-\ln 4 -\gamma}{\ln 4}\Bigg)/2 \\
		&= \Bigg(\frac12\log n+\left(\frac{2-\ln 4 -\gamma}{\ln 4}-\frac 12\right)\Bigg)/2 \\
		&\approx \Big(\frac12\log n - 0.474\ldots\Big)/2.
	\end{align*}
	Hence, by combining two summands, we obtain \eqref{eq:beta} as required.
	\end{proof}
	
	Together with~\eqref{eq:schemeB}, we obtain the desired upper bound for $\rho_B$ and obtain Theorem~\ref{thm:main}.
	Moreover, when $\sq=0$, we have that all words are good. In this case, we need not prepend $\bz$ with the two bits and thus, the average redundancy is simply given by $1+\frac{1}{2}\log {n}+ \beta\approx \frac{1}{2}\log {n}+0.526$.
	
\section{Balancing Schemes with Error-Correcting Capabilities}\label{sec:ecc}

%Story
%\begin{itemize}
%	\item $\sq=0$
%	\item Previous work.
%	\item Talk about Chee \etal.
%\end{itemize}

In this section, we set $\sq=0$ and consider variable-length balancing schemes that have certain error-correcting capabilities. Specifically, we propose Scheme~C and study methods to determine its average redundancy.

Now, the main ingredient of Scheme~C is what we termed as {\em  cyclic-balancing technique}, proposed by Chee \etal{} in \cite[Sec III-A]{chee2020}.
Here we provide a high-level description.
For a vector $\bx \in \{0,1\}^{n-1}$, let $\shift(\bx,i)$ be the vector obtained by cyclically shifting the components of $\bx$ to the right $i$ times. 
Similarly, $\shift(\bx,-i)$ is the vector obtained by cyclically shifts to the left $i$ times. 
To preserve certain distance properties, we always flip the first ``half'' of a message $\bx$, in other words, the first $m$ bits.
So, for simplicity, we denote $\flip(\bx,m)$ as $\flip_c(\bx)$.
Clearly, $\flip_c(\bx)$ may not have weight in $\{m-1,m\}$. 
Nevertheless, a key observation in~\cite{chee2020} is the following:
 if we consider all cyclic shifts of $\bx$, i.e., $\shift(\bx,i)$ for $i \in \bbracket{n}$, 
 then flipping the first $m$ bits of one of these shifts must yield a balanced word.
 In other words, $\flip_c(\shift(\bx,i))$ has weight in  $\{m-1,m\}$ for some $i\in\bbracket{n}$.
%\todo{Mention that because we always flip first half of the bits, we necessarily preserve the Hamming distance.} 
%\todo{cite Kumar Bhoi Singh. 
%A study of primer design with w-constacyclic shift over F4, Kumar, N., Bhoi, S. S., Singh, A. K. (2023). A study of primer design with w-constacyclic shift over F4. Theoretical Computer Science, 113925.}

With this observation, we can proceed in a similar fashion as Schemes~A and~B to design a variable-length balancing scheme.
Let $\C$ be an $[n-1,k,d]_2$-cyclic code.
In contrast with previous sections, the set of messages in this section belongs to $\C$ and we introduce the following notation:
\begin{align*}
	%\B_\C(n) \triangleq &\{\bc \in \mathbb{F}_2^{n-1}: \weight(\bc) \in \{m-1,m\} \text{, and} \\
	%&\flip_c(\shift(\bx,j))=\bc\, ,\, \text{ for } \bx \in \C, \text{ and }  j \in \bbracket{n-1}\}\,\\
	T_\C(\bx) &= \Big\{\, j \in \bbracket{n-1}: \flip_c(\shift(\bx,j)) \text{ has weight in } \{m-1,m\}\, \Big\} \text{ for } \bx \in \C\,.%\\
	%\tau_\C(\bx) &= \min T_\C(\bx)
\end{align*}

%\todo{define an $(n,\C,\rho_C)$-cyclic-balancing scheme. When $\C=\{0,1\}^n$, we obtain an $(n,n-1,\rho_C)$-balancing scheme.}

\noindent{\bf Scheme C}\\%: An $(n,\C;\rho_C)$-cyclic-balancing scheme\\
\hspace*{5mm} {\sc Input}: $\bx\in\C$\\[1mm]
\hspace*{1mm} {\sc Output}: $\bc' \in\B(n)$, $\bp \in \{0,1\}^*$\\[-5mm]
\begin{enumerate}[(I)]
	\item Determine the balanced word $\bc'$.
	\begin{itemize}
		\item $\tau \gets \min T_\C(\bx)$
		\item $\bc \gets \flip_c(\shift(\bx,\tau))$
		\item If $\weight(\bc)=m-1$, we append 1 to $\bc$ to obtain $\bc'$.
		\item If $\weight(\bc)=m$, we append 0 to $\bc$ to obtain $\bc'$.
	\end{itemize}
	\item Determine the prefix $\bp$.	
	\begin{itemize}
		\item Compute $\Gamma_\C(\bc)$ using Proposition~\ref{prop:GammaC}.
		\item $\bz \gets$ length-$r$ binary representation of the index of $\tau$ in $\Gamma_\C(\bc)$. Here, $r=\log|\Gamma_\C(\bc)|$.
	\end{itemize}
\end{enumerate}

%\todo{Say something about the rest of the section.}
The following result is similar to Theorem~\ref{thm:chee} from~\cite{chee2020}.
Here, we provide a proof for completeness.

\begin{proposition}\label{prop:cyclic}
Let $\C$ be an $[n-1,k,d]$-cyclic code.
Let $\C^*$ be the collection of all distinct codewords obtained from applying Scheme~C on words in $\C$. 
In other words, if $\enc(\bx)\triangleq (\bc',\bp)$ where $(\bc',\bp)$ is output of Scheme~C with input $\bx$, then we set 
\[\C^*\triangleq \{ \bc' : \enc(\bx)=(\bc',\bp) \text{~for some~}\bx\in\C \}\,.\] 
Then $\C^*$ is a $(n,d')$-balanced code where $d'=2\ceiling{d/2}$.
\end{proposition}

\begin{proof}
Since weight of $\bc'$ is $m$ after Step (I), we have that $\C^*$ is balanced. 

Hence, it remains to show the minimum distance of $\C^*$ is at least $d'$. 
Now, let $\bm$ denote the binary word with $m$ ones followed by $m-1$ zeroes.
Since $\C$ is cyclic, we have that $\shift(\bx,\tau)$ belongs to $\C$ and so, $\bc$ belongs to $\C+\bm$. In other words, the word $\bc$ in belongs to the coset $\C+\bm$.
Since all distinct words in $\C+\bm$ has minimum distance at least $d$, 
and all words in $\C^*$ has the same weight, we have that distinct words in $\C^*$ have distance at least $d'$.
\end{proof}

Note that $\enc$ is not an injective function in the statement of Proposition~\ref{prop:cyclic}. 
That is, it is possible for $\enc(\bx)=\enc(\bx')$ for distinct words $\bx,\bx'\in\C$.
Nevertheless, we can use the prefixes $\bp,\bp'$ for decoding purposes.
We provide an example in the next subsection.

A natural question is to determine the exact size of $\C^*$.
We describe a procedure to do so in Section~\ref{sec:rhoC-general}. 
Furthermore, Sections~\ref{sec:rhoC} to~\ref{sec:rhoC-general} are dedicated to the computation of the average redundancy of Scheme~C.

\subsection{Example with Error-Correction}

Fix $n=8$ and we consider the cyclic simplex $[7,3,4]$-code $\C$ that is generated by the polynomial $1+X^2+X^3+X^4$.
Here, we consider a systematic encoder, that is, a message $\bm$ is encoded to a simplex codeword $\bx\in\C$ by appending bits to $\bm$. For example, when $\bm=101$, the simplex encoder maps the message to $\bx=1011100$. Next, we apply Scheme~C onto $\bx$.
In Step (I), we find that the shift index $\tau$ is $1$ because $\bc=\flip_c(1011100,1)=1010 110$ has weight four.
So, here, we append 0 to $\bc$ to obtain the balanced word $\bc'$.
In Step (II), we determine $\Gamma_\C(\bc)$ to be $\{0,1,2,3\}$. Since we have four shift indices associated with $\bc$, we use a two-bit prefix $\bp=01$ to represent the shift index $\tau=1$.

In the table below, we write down the detailed steps for the encoding for all possible messages in $\{0,1\}^3$.
Note that here, we have that 
$\C^* = \{ 1111 0000,~1010 1100,~0110 0110,~00111010  \}$
and one can verify that $\C^*$ has distance four.

\begin{center}
\begin{tabular}{|C{20mm}|C{20mm}|C{12mm}|C{20mm}|c|c|}
\hline
message $\bm$ &
simplex codeword $\bx$ & 
shift index $\tau$ & 
balanced codeword $\bc'$ & 
$|\Gamma_\C(\bc)|$  & 
prefix $\bp$\\
\hline
000	& 0000000 &0  & 1111000{\bf 0} & 1 & -\\
101	& 1011100 &1  & 1010110{\bf 0} & 4 & 01\\
010	& 0101110 & 0 & 1010110{\bf 0} & 4 & 00\\
001	& 0010111 & 1 & 0110011{\bf 0} & 2 & 1 \\
100	& 1001011 & 0 & 0110011{\bf 0} & 2 & 0\\
110	& 1100101 & 0 & 0011101{\bf 0} & 1 & -\\
111	& 1110010 & 3 & 1010110{\bf 0} & 4 & 11\\ 
011	& 0111001 & 2 & 1010110{\bf 0} & 4 & 10\\
\hline
\end{tabular}
\end{center}

\noindent \textbf{Encoding}: 
To illustrate the error-correcting capability of this balancing scheme, 
we detail the steps to encode a data stream $000 101 1\ldots$ into balanced blocks of length $n=8$.
\begin{itemize}
	\item We pick the first three bits as the first packet: $000$. 
	We encode to simplex codeword $\bx=0000 000$ and Scheme C sets $\bc'=1111 0000$ with $\bp$ being the empty string.
	\item For the second packet, we pick the next three bits: $101$. 
	Repeating the procedure, we have $\bx=1011 100$ with $\bc'=1010 1100$ and $\bp=(01)$. 
	In this case, we prepend the next packet with $(01)$. Here, we place parenthesis to emphasize that the bits $(01)$ are redundant.
	\item In other words, the third packet is $(01)1$. Then we have $\bx=0111 001 $ with $\bc'= 1010 1100$ and $\bp=(10)$. Again, $(10)$ will be prepended to the fourth packet.
\end{itemize}

Therefore, the transmitted codewords are: $1111 0000$, $1010 1100$, $1010 1100$.

\noindent \textbf{Decoding:}
Observe that the transmitted codewords belong the extended simplex code which is a $[8,3,4]_2$-cyclic code.
Therefore, assuming that each block has at most one error, we are able to correct the error and so the balanced codewords are received correctly.

We invert the encoding process and attempt to recover the data stream. 
Here, we assume that we have decoded the fourth codeword and obtained the fourth data packet.
\begin{itemize}
	\item The third codeword is $\bc'=1010 1100$ or $\bc=1010 110$. 
	Since $\Gamma_\C(\bc)$ has four indices, we pick the first two bits of the fourth packet which is $(10)$. 
	As $(10)$ corresponds to shift index two, we set $\bx=\shift(\flip_c(\bc),-2)=0111001$. Hence, the third packet is $101$.
	\item The second codeword is $\bc'=1010110$. Hence, $\bc=1010 110$ and $|\Gamma_\C(\bc)|=4$. 
	So, we pick the first two bits of the third packet $(01)$.
	Then $\bx=\shift(\flip_c(\bc),-1)=101 1100$ and the second packet is $101$.
	\item Finally, the first codeword is $\bc=1111 0000$ and so, $\bc=1111 0000$ and $|\Gamma_\C(\bc)|=1$.
	Therefore, the shift index is zero and $\bx=0000 000$ and  the first packet $000$.
	\item Hence, the packets are $000$, $101$, $(01)1$, \ldots. Recall that when we decoded the second codeword, we figured out the corresponding prefix is $(01)$. Therefore, we remove the prefix and so, the data stream is $000 101 1\ldots $.
\end{itemize}

\subsection{Average Redundancy Analysis}\label{sec:rhoC}

In this subsection, we determine the average redundancy $\rho_\C$.
As with previous sections, we consider the set of balancing indices that can be received with a codeword $\bc$.
Specifically, we consider the following set of indices:
\begin{equation*}
\Gamma_\C(\bc) \triangleq \Big\{j \in \bbracket{n-1}: \flip_c(\shift(\bx,j))=\bc\, ,\, \min T_\C(\bx)=j \text{ for some } \bx \in \C \Big\}\,.
\end{equation*}

Then we have two tasks. First, we need a simple characterization of $\Gamma_\C(\bc)$ and we do so in Proposition~\ref{prop:GammaC}.
Unlike previous schemes, as our messages are obtained from a cyclic code $\C$, we have an additional task to determine the possible received words $\bc$. 
Specifically, we need to determine 
\begin{equation*}
	\B(\C) \triangleq \Big\{\flip_c(\bx) : \bx\in\C \text{ and } \flip_c(\bx) \text{ has weight in } \{m-1,m\} \Big\}\,.
\end{equation*}

Then proceeding as before, we define $\gamma_\C(i,n-1)\triangleq |\{\bc\in \B(\C) : |\Gamma_\C(\bc )| = i \}|$ and 
we have that the average redundancy to be
\begin{equation*}
	\rho_\C = 1+\frac{1}{|\C|}\sum_{i=1}^m i \gamma_\C(i,n-1) \log i .
\end{equation*}

Observe that the domain of $i$ is between 1 and $m$. This is valid because the indices from $\Gamma_\C(\bc)$ belong to $\bbracket{m-1}$. This follows from a result in Chee \etal{}~\cite{chee2020}. 
For completeness, we restate the lemma and provide a simple proof.

\begin{lemma}\label{lem:maxTC}
	Let $\bx\in\C$ and set $\tau = \min T_\C(\bx)$. Then $\tau\leq m-1$.
\end{lemma}
\begin{proof}
Suppose otherwise that $\tau\ge m$.
We claim that $\tau-m$ belongs to $T_\C(\bx)$.
Let $\bc=\flip_c(\shift(\bx,\tau))$ and so, $\weight(\bc)\in\{m-1,m\}$.
Then $\flip_c(\shift(\bx,\tau-m))=\flip(\bc,n)$ has weight in $\{m-1,m\}$.
Therefore, $\tau-m$ belongs to $\min \{T_\C(\bx)\}$ but $\tau-m<\tau$, contradicting the minimality of $\tau$.
\end{proof}

For our first task of characterizing $\Gamma_\C(\bc)$, we introduce the notion of cyclic running sum.

\begin{definition}
	Let $\bc=c_1c_2\ldots c_{n-1} \in \{0,1\}^{n-1}$. The \textit{cyclic running sum} of $\bc$, denoted by $CR(\bc)$, is a length-$m$ integer-valued vector indexed by $\bbracket{m-1}$ defined by 
	\begin{equation*}
		CR(\bc)_i = 
		\begin{cases}
			0, &\text{if } i = 0, \\
			CR(\bc)_{i-1}+ (-1)^{c_i+1}+(-1)^{c_{i+m}+1}, &\text{if } i > 0.
		\end{cases}
	\end{equation*}
\end{definition}

Then we obtain the following characterization.
\begin{proposition}\label{prop:GammaC}
	Let $\bc$ belong to $\B(\C)$.
	\begin{equation}
		\Gamma_\C(\bc) = \{0\} \cup\{ i \in\{1,\ldots, m-1\}: CR(\bc,j) \ne 0 \text{ for all } 1\le j \le i \} . \label{eq:GammaC}
	\end{equation}
\end{proposition}

\begin{proof} Let $i$ belong to the set on the RHS of \eqref{eq:GammaC}.
	Set $\bx=\shift(\flip_c(\bc),-i)$ and we claim that $i=\min T_\C(\bx)$.
	Clearly, $i\in T_\C(\bx)$ and so, it remains to show that $i$ is the smallest index.
	
	For $j<i$, we consider the weight of the word $\flip_c(\shift(\bx,j))$.
	Observe that $\flip_c(\shift(\bx,j))=\flip_c(\shift(\flip_c(\bc),-t))$ where $t=i-j$.
	Since $\shift(\flip_c(\bc),-t)= \overline{c_{t+1}\cdots c_{m}}c_{m+1}\cdots c_{2m-1}\overline{c_1\cdots c_{t}}$,
	we have that \[\flip_c(\shift(\flip_c(\bc),-t))= c_{t+1}\cdots c_{m}\overline{c_{m+1}\cdots c_{t+m}}c_{t+m+1}\cdots c_{2m-1}\overline{c_1\cdots c_t}.\] 
	Therefore, $\weight(\flip_c(\shift(\flip_c(\bc),-t))) = \sum_{i=1}^{2m-1} c_i - 2CR(\bc)_t=\weight(\bc)- 2CR(\bc)_t$.
	Since $0<t<i$, we have that $2CR(\bc)_t\ne 0$ and so,  $\weight(\flip_c(\shift(\flip_c(\bc),-t)))\notin\{m-1,m\}$.
	In other words, $i=\min T_\C(\bx)$.
\end{proof}

In the next two subsections, we proceed to our next task of determining $\gamma_\C(i,n-1)$.
Unfortunately, in this case, we are unable to obtain a closed formula for this quantity for general cyclic code $\C$.
Nevertheless, in the case when $\C=\{0,1\}^{n-1}$, that is, $\C$ comprise all binary words of length $n-1$, we have a closed formula.
On the other hand, for a general cyclic code $\C$, we provide a method that computes $\gamma_\C(i,n-1)$ in polynomial time under some conditions.  
%\todo{Say something}

\subsection{Average Redundancy Analysis when $\C=\{0,1\}^{n-1}$}

In this subsection, we set  $\C=\{0,1\}^{n-1}$ and we prove the following theorem.

\begin{theorem}\label{thm:gammaC-all}
For $1\le i\le m$, we have that
\begin{equation*}
\gamma_\C(i,n) =  \frac{2}{i}\binom{2i-2}{i-1} \binom{2m-2i}{m-i}\,.
%\begin{cases}
%	\frac{4}{i}\binom{2i-2}{i-1} \binom{2m-1-2i}{m-i}, & \mbox{if $1\le i\le m-1$},\\
%	\frac{2}{i}\binom{2i-2}{i-1}, & \mbox{if $i = m$ }.
%\end{cases}
\end{equation*}
\end{theorem}

\begin{proof}
We first consider the case where $i<m$.
Suppose that $|\Gamma_\C(\bc)|=i$ and we set $CR(\bc)=(s_0,s_1,\ldots, s_{m-1})$. 
Then Proposition~\ref{prop:GammaC} states that $CR(\bc)$ satisfy the following properties:
\begin{enumerate}[(S1)]
	\item $s_0 = s_i = 0$.
	\item $s_j > 0$ for all $0<j<i$ or $s_j > 0$ for all $0<j<i$.
\end{enumerate}
Therefore, $\gamma_\C(i,n)$ is the number of binary length $n-1$ words whose cyclic running sums satisfy (S1) to (S3).

First, we observe that the values of $c_1,c_2,\ldots, c_i,c_{m+1},c_{m+2}\ldots, c_{m+i}$ completely determines the values of $s_1,s_2,\ldots, s_i$. So, our first claim is the following:
\begin{quote}
{\bf (Claim 1)}. The number of choices for $c_1,\ldots ,c_i,c_{m+1},\ldots, c_{m+i}$ that result in $s_1,\ldots, s_i$ satisfying (S1) and (S2) is $\frac{2}{i}\binom{2i-2}{i-1}$.
\end{quote}

To prove Claim~1, we use the lattice-path results in Section~\ref{sec:lattice}.
Specifically, we consider the length-$2i$ word $\bp=c_1c_{m+1}c_2c_{m+2}\ldots c_{i}c_{m+i}$ and the corresponding lattice path $\Pi(\bp)$. Since $s_i=0$, we have that $\weight(\bp)=i$ and $\Pi(\bp)$ starts from $(0,0)$ and ends at $(i,i)$. 

When $c_1 = 1$, we have that $s_j>0$ for all $0<j<i$. Hence, $\Pi(\bp)$ travels from $(0,0)$ to $(0,1)$, stays above the line $Y=X+1$ until it reaches $(i-1,i)$, and finally ends $(i,i)$. Applying Theorem~\ref{thm:lattice} with $\Delta_x=\Delta_y=i-1$ and $d-a=i-1$, we have that the number of such paths is $\frac1i\binom{2i-2}{i-1}$. When $c_1=0$, a similar argument applies that we have the number of such paths is also $\frac1i\binom{2i-2}{i-1}$.

Next, we determine the number of choices for $c_{i+1},\ldots, c_{m-1},c_m,c_{m+i+1},\ldots, c_{2m-1}$.
Here, the number of remaining bits is $2m-1-2i$. Among them, the number of ones is $m-i-1$ or $m-i$. Therefore, the number of choices is $\binom{2m-1-2i}{m-i} + \binom{2m-1-2i}{m-i-1}= \binom{2m-2i}{m-i}$.
Therefore, combining with Claim 1, we have the theorem.

When $i=m$, we append $c_{2m}$ so that the length-$2m$ word $\bc'$ has weight $m$.
Then we proceed as before and consider the lattice path $\Pi(\bc')$.
We can argue that the lattice path $\Pi(\bc')$ starts from $(0,0)$ and ends at $(m,m)$.
Furthermore, $\Pi(\bc')$ either always stays above $Y=X$ or below $Y=X$. In other words, the number of words/paths is exactly $\frac2i\binom{2i-2}{i-1}$, as required.
\end{proof}

\begin{remark}
In this proof, for the sake of consistency, we use lattice-path combinatorics to obtain Claim 1. 
Here, we sketch an alternative (and essentially, equivalent) derivation that makes use of the combinatorial object, Dyck words.
A {\em Dyck word} $\bx$ of length $2n$ is a binary word with weight $n$ where all prefixes of $\bx$ have more ones than zeroes.
It is well-known that the number of length-$2n$ Dyck word is given by the $n$-th Catalan number $C_n=\frac{1}{n+1}\binom{2n}{n}$
(see for example, \cite[Chapter 6]{sedgewick2013introduction}). 
We can then argue that $c_2c_{m+2}\cdots c_i$ is a Dyck word of length $2i-2$ and this recovers Claim 1.
\end{remark}

Therefore, we have the following expression for the average redundancy.
\begin{corollary}
$\rho_\C = 1+\frac{1}{2^{2m-1}}\sum_{i=1}^{m}  2 \binom{2i-2}{i-1}\binom{2m-2i}{m-i} \log i$.
\end{corollary}

The next proposition provides an asymptotic estimate for $\rho_\C$.

\begin{proposition}\label{prop:rhoC}
If $\C=\{0,1\}^{n-1}$, then $\rho_\C\sim \log n$.
\end{proposition}

\begin{proof}
First, since $\log i\le \log m$ for all $i$, we have that $\rho_\C\le 1+\log m = \log n$.

Next, we fix $0<\epsilon,\delta<1$ and show that 
\begin{equation}\label{eq:rhoC-lowbound}
\rho_\C \gtrsim	 1+ (1-\epsilon) \left(-2+ (\log n - 1)\left(1 - \frac{2\sin^{-1}\sqrt{\delta}}{\pi}\right)\right)
%\rho_\C \gtrsim 1+(1-\epsilon)( (\log n - 3) - 2\delta\log n )\,.
\end{equation}
In other words, we have that $\rho_\C \gtrsim \log n - 2\sim \log n$, as required.

To show \eqref{eq:rhoC-lowbound}, we use Stirling's approximation  $\binom{2N}{N} \sim \frac{2^{2N}}{\sqrt{\pi N}}$ (see for example, \cite{sedgewick2013introduction}).
In other words, for any $\epsilon > 0$, there exists $N_0$ such that 
\begin{equation}\label{stirling}
	\binom{2N}{N} \geq (1-\epsilon)\frac{2^{2N}}{\sqrt{\pi N}} \text{ for $N\ge N_0$} \,.
\end{equation}
Given $\delta>0$, we then choose $m_0$ such that $\delta m_0\ge N_0 + 1$. 
Hence, for all $m\ge m_0$ and $\delta m\le i\le (1-\delta) m$, we have that $i-1 \ge \delta m_0-1\ge N_0$, and so, 
$\binom{2i-2}{i-1} \geq (1-\epsilon)\frac{2^{2i-2}}{\sqrt{\pi (i-1)}}$\,. 
Similarly, $m-i\ge \delta m \ge \delta m_0 \ge N_0$ and so,
$\binom{2m-2i}{m-i} \geq (1-\epsilon)\frac{2^{2m-2i}}{\sqrt{\pi (m-i)}}$\,.
Therefore, we have that 
\[ \sum_{i=1}^m 2\log(i) \binom{2i-2}{i-1}\binom{2m-2i}{m-i} 
\ge \sum_{i=\delta m}^{(1-\delta) m} 2\log(i) \Big(1-\epsilon\Big) \frac{2^{2m-2}}{\pi\sqrt{ (i-1)(m-i)}}
=  \frac{2^{2m-1}(1-\epsilon)}{\pi} \sum_{i=\delta m}^{(1-\delta) m} \frac{\log i}{\sqrt{ (i-1)(m-i)}}\,.\] 

Now, we estimate the sum using an integral. That is,
\begin{align*}
\sum_{i=\delta m}^{(1-\delta) m} \frac{\log i}{\sqrt{ (i-1)(m-i)}} 
	& =\sum_{i=\delta m}^{(1-\delta)m} \frac{\log i}{m \sqrt{(i-1)/m (1 - i /m)}}\\
	&\sim \int_{\delta}^{1-\delta} \frac{\log(\lambda m)}{\sqrt{\lambda(1-\lambda)}} \,d\lambda \\
	&= \int_{\delta}^{1-\delta} \frac{\log m}{\sqrt{\lambda(1-\lambda)}} \,d\lambda +\int_{\delta}^{1-\delta} \frac{\log \lambda}{\sqrt{\lambda(1-\lambda)}} \,d\lambda\\
	&\ge \int_{\delta}^{1-\delta} \frac{\log m}{\sqrt{\lambda(1-\lambda)}} \,d\lambda +\int_{0}^{1} \frac{\log \lambda}{\sqrt{\lambda(1-\lambda)}} \,d\lambda\,.
%	&\ge \int_{0}^{1} \frac{\log(\lambda m)}{\sqrt{\lambda(1-\lambda)}} \,d\lambda - 2\delta \frac{\log m \binom{2m-2}{m-1}}{2^{2m-1}}    \\
%	& = \int_{0}^{1} \frac{\log(\lambda)}{\sqrt{\lambda(1-\lambda)}} \,d\lambda + \int_{0}^{1} \frac{\log(m)}{\sqrt{\lambda(1-\lambda)}} \,d\lambda- 2\delta\log m 
\end{align*}

Now, the last inequality follows from the fact that $\log \lambda$ is negative in the interval $(0,1)$ and we have that $\int_{0}^{1} \frac{\log \lambda}{\sqrt{\lambda(1-\lambda)}} \,d\lambda = -2\pi$.
Moreover, $\int_{\delta}^{1-\delta} \frac{\log m }{\sqrt{\lambda(1-\lambda)}} \,d\lambda= (\pi-2\sin^{-1}\sqrt{\delta})\log m$, we can now complete our estimate for $\rho_\C$.

\begin{align*}
	\rho_\C & \ge 1+\frac{1}{2^{2m-1}}\sum_{i=1}^{m} 2\log(i) \binom{2i-2}{i-1}\binom{2m-2i}{m-i}\\
	&\gtrsim 1+ \frac{2^{2m-1}(1-\epsilon)}{2^{2m-1}\pi} (-2\pi+(\pi-2\sin^{-1}\sqrt{\delta})\log m)\\
	&\ge 1+ (1-\epsilon) \left(-2+ (\log n - 1)\left(1 - \frac{2\sin^{-1}\sqrt{\delta}}{\pi}\right)\right)\,.
\end{align*}
This completes the proof of \eqref{eq:rhoC-lowbound}.
\end{proof}

\subsection{Average Redundancy Analysis for General Cyclic Codes $\C$}\label{sec:rhoC-general}
For a general $[n-1,k,d]$-cyclic $\C$, we propose a trellis that determines whether a word $\bc$ belongs to $\B(\C)$.
%sequential time-variant machine to verify whether a word $\bc$ belongs to $\B_\C$ or not. 
Moreover, it provides a polynomial-time method to compute the average redundancy. %$\rho_\C$ for any cyclic code $\C$. 
Let $\bc = c_1c_2 \ldots c_{n-1}$ be a length-$(n-1)$ string.
To determine if $\bc$ belongs to $\B(\C)$, we start from the root and traverse the trellis according the bits at the $i$-th and $(i+m)$-th positions, where  $1\leq i \leq m-1$.
Specifically, at the $i$th step, we look at the label $c_ic_{i+m}$ and transition accordingly. In the final step, we move according to $m$-th position, that is, $c_m$.

We provide a formal description of the trellis.
Let $\bH = [\bh_1, \bh_2, \ldots, \bh_n]$ the $(n-k)\times n$-parity check matrix of the cyclic code $\C$. Here, $\bh_i$ denotes the $i$-th column of $\bH$.
%
%Let $\bc = c_1c_2 \ldots c_{n-1}$ where $n-1=2m-1$ and $c_i \in \{0,1\}$. The machine starts with an empty word and generates all words $\bc \in \B_\C$ by adding 2 bits at $i$-th and $(i+m)$-th positions for $1\leq i \leq m-1$ sequentially. A final step is to add $0$ or $1$ at the $m$-th position. Let $H_\C = [h_1, h_2, \ldots, h_n]$ of size $k\times n$ be the parity check matrix of the code $\C$, where $h_i$ denotes the $i$-th column of $H_\C$.
%
Then each vertex or state $\bu$ of the trellis stores five values. 
Specifically, we describe the state $\bu \triangleq (\level_u, CR_\bu, \weight_\bu, \synd_\bu, \balind_\bu )$ as follow. Here, the five values correspond to the level, cyclic running sum, weight, syndrome, and number of balancing indexes, respectively.
To describe the transition rules, let $\bv = (\level_\bv, CR_\bv, \weight_\bv, \synd_\bv, \balind_\bv )$ be a state following $\bu$, or $\bu \ra \bv$, with the edge labelled by $c_ic_{i+m}$ or $c_m$.
\begin{itemize}
	\item The machine has $m+1$ levels. The $\level$ keeps track of the current level. Hence, $\level_\bv = \level_\bu + 1$. %Note that each transition from level $i-1$ to level $i$ corresponds to assigning values to $c_i$ and $c_{i+m}$.
	\item At level $0$, $CR$ is initialized to be $0$. For $1\le i\le m-1$, transition from level $i-1$ to $i$ increases $CR$ by $(-1)^{c_i+1}+ (-1)^{c_{i+m}+1}$. From level $m-1$ to $m$, the value $CR$ increases by $(-1)^{c_m+1}$.
	\item At level $0$, $\weight$ is initialized to be $0$. For $1\le i\le m-1$, transition from level $i-1$ to $i$ increases $\weight$ by $c_i+ c_{i+m}$. From level $m-1$ to $m$, the value $\weight$ increases by ${c_m}$.
	\item At level $0$, $\synd$ is initialized to be zero-vector $\boldsymbol{0}$ of length $n-k$. For $1\le i\le m-1$, transition from level $i-1$ to $i$ modifies $\synd$ by setting $\synd\gets \synd + (c_{i}+1)\bh_i + c_{i+m} \bh_{i+m}$. From level $m-1$ to $m$, we have $\synd\gets \synd + (c_{m}+1)\bh_m$.
	\item The $\balind$ keeps track of the number of balancing indices. Here, $\balind$ can be either take a integer value from $1$ to $m$ or the symbol $\mathord{?}$. The starting state has $\balind=0$. The transition rules are as follow. %Details of index transition are discussed in {\bf Index Transition}.
	\begin{itemize}
		\item If $\balind_\bu = 0$ and $c_1c_{m+1} \in \{00,11\}$, then $\balind_\bv = \mathord{?}$. If $\balind_\bu = 0$ and $c_1c_{m+1} \in \{10,01\}$, then $\balind_\bv = 1$.
		%\item If $\balind_\bu = \alpha$ is a positive integer, then  $\balind_\bv = \alpha$.
		\item If $\balind_\bu \ne \mathord{?}$, then  $\balind_\bv = \balind_\bv$.
		\item If $\balind_\bu = \mathord{?}$ and $CR_\bv \neq 0$, then  $\balind_\bv = \mathord{?}$. If $\balind_\bu = \mathord{?}$ and $CR_\bv = 0$, then  $\balind_\bv = \level_\bv$.
	\end{itemize}
\end{itemize}

\begin{comment}
 Each state at level $0\le i \le m-2$ has four possible outputs, labeled by $c_ic_{i+m} \in \{00,01,10,11\}$. We note the difference in transition from $(m-1)$-th level to the last level. Each state at level $m-1$ has only two possible outputs, labeled by $0$ or $1$. If $\bu$ moves to $\bv$ along the $0$-edge, $\level_\bv = m$, $CR_\bv = CR_\bu-2$, $\weight_\bv= \weight_\bu$, $\synd_\bv= \synd_\bu+h_m$, and $\balind_\bv$ follows the {\bf Index Transition}. Similarly, if $\bu$ moves to $\bv$ along the $1$-edge, $\level_\bv = m$, $CR_\bv = CR_\bu+2$, $\weight_\bv= \weight_\bu+1$, $\synd_\bv= \synd_\bu$, and $\balind_\bv$ follows the {\bf Index Transition}.

{\centering{\bf{Index Transition}}
%From here, we describe transition rules. The transitions of $\level, CR, \weight$, and $\synd$ have been mentioned above. Here is the detail of $\balind$ transition from $\bu \ra \bv$, where $\bv = (\level+1, CR_\bv, \weight_\bv, \synd_\bv, \balind_\bv )$. 
\begin{itemize}
	\item If $\balind_\bu = 0$ and $c_1c_{m+1} \in \{00,11\}$, then $\balind_\bv = \mathord{?}$. If $\balind_\bu = 0$ and $c_1c_{m+1} \in \{10,01\}$, then $\balind_\bv = 1$.
	\item If $\balind_\bu = \alpha$ is a positive integer, then  $\balind_\bv = \alpha$.
	\item If $\balind_\bu = \mathord{?}$ and $CR_\bv \neq 0$, then  $\balind_\bv = \mathord{?}$. If $\balind_\bu = \mathord{?}$ and $CR_\bv = 0$, then  $\balind_\bv = \level_\bv$.
\end{itemize}}
\end{comment}

Let $\boldsymbol{0}_{n-k}$ be the zero vector of length $n-k$.
Then the root of the trellis or the starting state is labelled with $(0, 0, 0, \boldsymbol{0}_{n-k},0)$. At level $m$, there are $2m$ final states with labels:
$(m,0,m,\boldsymbol{0}_{n-k},j)$ and $(m,0,m-1,\boldsymbol{0}_{n-k},j)$ for $1 \le j \le m$. We observe that a sequence of states or a path from starting state to one of the final states uniquely determines a codeword in $\B(\C)$, and vice versa. 
Moreover, for $1\leq j \leq m$, a path from $(0, 0, 0, \boldsymbol{0}_{n-k},0)$ to   $(m,0,w,\boldsymbol{0}_{n-k},j)$ corresponds to a word $\bc$ in $\B(\C)$ of weight $w$ with $|\Gamma_\C(\bc)|=j$. Therefore, the total number of paths from $(0, 0, 0, \boldsymbol{0}_{n-k},0)$ to either $(m,0,m-1,\boldsymbol{0}_{n-k},j)$ or $(m,0,m,\boldsymbol{0}_{n-k},j)$ yields the quantity $\gamma_\C(j,n)$.

\begin{example}\label{exa:trellis-1}
Let $\C$ be the cyclic Hamming $[7,4,3]$-code with parity-check matrix
$\bH = \begin{bmatrix}
	1&0&1&1&1&0&0 \\ 0&1&0&1&1&1&0 \\ 0&0&1&0&1&1&1
\end{bmatrix} $. 
In Figure~\ref{fig:743cyclic}, we illustrate the trellis corresponding to $\C$. 
The path highlighted in bold corresponds to the codeword $\bc=1010110$ and we check that the weight of $\bc$ is four and $|\Gamma_\C(\bc)|=4$. This is reflected in the third and fifth coordinates in the final state $(4,0,4,\boldsymbol{0}_3,4)$.
%how to get the same results using state-machine. For space constraint, we ignore states without paths to final states. We also ignore 2 final states with $\balind=3$, which are unreachable from the starting state.
%By exhaustive search, we observe that $\B_\C$ has $8$ codewords of weight $3$ or $4$. We also have $\gamma_\C(1,7) = 4$, $\gamma_\C(2,7) = 2$, $\gamma_\C(3,7) = 0$, $\gamma_\C(4,7) = 2$. These results are listed in the table \ref{tab:743cyclic}.

For completeness, we list all possible encodings of words in $\C$ in the table below. We verify that the final states of the corresponding paths provide the correct values of $\weight$ and $\balind$.
 
\begin{center}
	\begin{tabular}{|C{20mm}|C{12mm}|C{20mm}|C{20mm}|c|}
		\hline
		Hamming codeword $\bx$ & 
		shift index $\tau$  & 
		word $\bc$ of weight $3$ or $4$ & 
		Final State \\
		\hline
		%0000 & 
		0000000 &0  & 1111000 & $(4,0,4,{\boldsymbol{0}_3},1)$ \\
		%1000	& 
		1100101 & 0 & 0011101 & $(4,0,4,{\boldsymbol{0}_3},1)$ \\
		\hline
		%1111	& 
		1111111 & 0 & 0000111 & $(4,0,3,{\boldsymbol{0}_3},1)$ \\
		%0111	& 
		0011010 & 0 & 1100010 & $(4,0,3,{\boldsymbol{0}_3},1)$ \\
		\hline
		%0101	& 
		1001011 & 0 & 0110011 & $(4,0,4,{\boldsymbol{0}_3},2)$ \\ 
		%1110	& 
		0010111 & 1 & 0110011 & $(4,0,4,{\boldsymbol{0}_3},2)$ \\
		\hline
		%1010	& 
		0110100 & 0 & 1001100 & $(4,0,3,{\boldsymbol{0}_3},2)$ \\
		%0001	& 
		1101000 &1  & 1001100 & $(4,0,3,{\boldsymbol{0}_3},2)$  \\
		\hline
		%1101	& 
		0101110 & 0 & 1010110 & $(4,0,4,{\boldsymbol{0}_3},4)$ \\
		%1011	& 
		1011100 & 1 & 1010110 & $(4,0,4,{\boldsymbol{0}_3},4)$ \\
		%0011	& 
		0111001 & 2 & 1010110 & $(4,0,4,{\boldsymbol{0}_3},4)$ \\
		%0110	& 
		1110010 & 3 & 1010110 & $(4,0,4,{\boldsymbol{0}_3},4)$ \\
		\hline
		%0010	& 
		1010001 & 0 & 0101001 & $(4,0,3,{\boldsymbol{0}_3},4)$ \\
		%0100	& 
		0100011 & 1 & 0101001 & $(4,0,3,{\boldsymbol{0}_3},4)$  \\
		%1100	& 
		1000110 & 2 & 0101001 & $(4,0,3,{\boldsymbol{0}_3},4)$ \\
		%1001	& 
		0001101 & 3 & 0101001 & $(4,0,3,{\boldsymbol{0}_3},4)$ \\
		\hline
	\end{tabular}
\end{center}
\end{example}

Henceforth, we use  $T(\bv)= T(\level_\bv, CR_\bv, \weight_\bv, \synd_\bv, \balind_\bv)$ to denote the number of paths from the root $(0, 0, 0, {\boldsymbol{0}_k},0)$ to the vertex $\bv$. 
Following the above discussion, determining $\gamma_\C(j,n)$ is equivalent to computing $T(m,0,m-1,\boldsymbol{0}_{n-k},j)+T(m,0,m,\boldsymbol{0}_{n-k},j)$.

To do so, we make use of the following three propositions.
We omit the detailed proofs as they are immediate from the transition rules used to describe the trellis.

\begin{figure*}[t!]
	\centering
	\begin{tikzpicture}[shorten >=1pt,node distance=1.0cm,on grid,auto, 
		block/.style={
			draw,
			fill=white,
			rectangle,
			minimum width=2.5cm,
			font=\footnotesize}] % Some customizations related to the size and the discatnce between nodes and arrow heads
		\node[block] (Start) [] {$0,0,0, {\boldsymbol{0}_3},0$}; % Here the nodes and coordinates are defined
		\node[block] (00) [right=of Start, xshift=2.5cm, yshift=3cm]   {$1,-2,0,[1,0,0], \mathord{?}$};
		\node[block] (01) [right=of Start, xshift=2.5cm, yshift=1cm]   {$1,0,1,[0,1,1], 1$};
		\node[block] (10) [right=of Start, xshift=2.5cm, yshift=-1cm]   {$1,0,1,[0,0,0], 1$};
		\node[block] (11) [right=of Start, xshift=2.5cm, yshift=-3cm]   {$1,2,2,[1,1,1], \mathord{?}$};

		\path[->,color=red] (Start) edge [above] node  {} (00);
		\path[->,color=orange] (Start.north east) edge [above] node  {} (01.south west);
		\path[->,color=green] (Start.south east) edge [above] node  {} (10.north west);
		\path[->,color=blue, line width=0.5mm] (Start) edge [above] node  {} (11);
		
		%\path[->,color=black] (Start) edge [above] node  {00} (00);
		%\path[->,color=black] (Start.north east) edge [above] node  {01} (01.south west);
		%\path[->,color=black] (Start.south east) edge [above] node  {10} (10.north west);
		%\path[->,color=black] (Start) edge [above] node  {11} (11);
		
		\node[block] (0010) [right=of 00, xshift=2.5cm, yshift=1cm]   {$2,-2,1,[1,0,0], \mathord{?}$};
		\node[block] (0011) [right=of 00, xshift=2.5cm, yshift=0cm]   {$2,0,2,[1,1,1], 2$};
		%\path[->,color=black] (00.north east) edge [above] node  {10} (0010.south west);
		%\path[->,color=black] (00) edge [above] node  {11} (0011);
		\path[->,color=green] (00.north east) edge [above] node  {} (0010.south west);
		\path[->,color=blue] (00) edge [above] node  {} (0011);
		
		\node[block] (0100) [right=of 01, xshift=2.5cm, yshift=1cm]   {$2,-2,1,[0,0,1], 1$};
		\node[block] (0101) [right=of 01, xshift=2.5cm, yshift=0cm]   {$2,0,2,[0,1,0], 1$};
		\path[->,color=red] (01.north east) edge [above] node  {} (0100.south west);
		\path[->,color=orange] (01) edge [above] node  {} (0101);
		%\path[->,color=black] (01.north east) edge [above] node  {00} (0100.south west);
		%\path[->,color=black] (01) edge [above] node  {01} (0101);
		
		\node[block] (1010) [right=of 10, xshift=2.5cm, yshift=0cm]   {$2,0,2,[0,0,0], 1$};
		\node[block] (1011) [right=of 10, xshift=2.5cm, yshift=-1cm]   {$2,2,3,[0,1,1], 1$};
		%\path[->,color=black] (10) edge [above] node  {10} (1010);
		%\path[->,color=black] (10.south east) edge [above] node  {11} (1011.north west);
		\path[->,color=green] (10) edge [above] node  {} (1010);
		\path[->,color=blue] (10.south east) edge [above] node  {} (1011.north west);
		
		\node[block] (1100) [right=of 11, xshift=2.5cm, yshift=0cm]   {$2,0,2,[1,0,1],2$};
		\node[block] (1101) [right=of 11, xshift=2.5cm, yshift=-1cm]   {$2,2,3,[1,1,0], \mathord{?}$};
		\path[->,color=red] (11) edge [above] node  {} (1100);
		\path[->,color=orange, line width=0.5mm] (11.south east) edge [above] node  {} (1101.north west);
		%\path[->,color=black] (11) edge [above] node  {00} (1100);
		%\path[->,color=black] (11.south east) edge [above] node  {01} (1101.north west);
		
		\node[block] (32?) [right=of Start, xshift=9.5cm, yshift=4cm]   {$3,-2,2,[0,0,0], \mathord{?}$};
		%\node[block] (3031) [right=of Start, xshift=12cm, yshift=2cm]   {$3,0,3,[0,0,0], 1$};
		\node[block] (3032) [right=of Start, xshift=9.5cm, yshift=-1cm]   {$3,0,3,[0,0,0], 1$};
		\node[block] (324?) [right=of Start, xshift=9.5cm, yshift=-4cm]   {$3,2,4,[1,1,0], \mathord{?}$};
		\node[block] (3242) [right=of Start, xshift=9.5cm, yshift=3cm]   {$3,2,4,[1,1,0], 2$};
		\node[block] (3131) [right=of Start, xshift=9.5cm, yshift=1cm]   {$3,0,3,[1,1,0], 1$};
		%\node[block] (3033) [right=of Start, xshift=12cm, yshift=-2cm]   {$3,0,3,[1,1,0], 1$};
		\node[block] (322) [right=of Start, xshift=9.5cm, yshift=-3cm]   {$3,-2,2,[0,0,0], 2$};
		\path[->,color=orange] (0010) edge [above] node  {} (32?);
		\path[->,color=blue] (0011) edge [above] node  {} (3242);
		\path[->,color=blue] (0100.south east) edge [above] node  {} (3032.north west); %11
		\path[->,color=orange] (0101) edge [above] node  {} (3131); %01
		\path[->,color=green] (1010) edge [above] node  {} (3032); %10
		\path[->,color=red] (1011.north east) edge [above] node  {} (3131.south west); %00
		\path[->,color=red] (1100) edge [above] node  {} (322);
		\path[->,color=green,line width=0.5mm] (1101) edge [above] node  {} (324?);
		%\path[->,color=black] (0010) edge [above] node  {01} (32?);
		%\path[->,color=black] (0011) edge [above] node  {01} (3242);
		\path[->,color=blue] (0100.south east) edge [above] node  {} (3032.north west); %11
		%\path[->,color=black] (0101) edge [above] node  {} (3131); %01
		%\path[->,color=black] (1010) edge [above] node  {} (3032); %10
		%\path[->,color=black] (1011.north east) edge [above] node  {} (3131.south west); %00
		%\path[->,color=black] (1100) edge [above] node  {00} (322);
		%\path[->,color=black] (1101) edge [above] node  {10} (324?);
		%\path[->,color=black] ([xshift=0mm]0011.south east) edge [above] node  {} ([xshift=0mm]33.north west);
		%\path[->,color=black] (1101) edge [above] node  {} (33);

		\node[block] (434) [right=of Start, xshift=13cm, yshift=4cm]   {$4,0,3,{\boldsymbol{0}_3}, 4$};
		\node[block] (442) [right=of Start, xshift=13cm, yshift=3cm]   {$4,0,4,{\boldsymbol{0}_3}, 2$};
		\node[block] (431) [right=of Start, xshift=13cm, yshift=1cm]   {$4,0,3,{\boldsymbol{0}_3}, 1$};
		\node[block] (441) [right=of Start, xshift=13cm, yshift=-1cm]   {$4,0,4,{\boldsymbol{0}_3}, 1$};
		\node[block] (432) [right=of Start, xshift=13cm, yshift=-3cm]   {$4,0,3,{\boldsymbol{0}_3}, 2$};
		\node[block] (444) [right=of Start, xshift=13cm, yshift=-4cm]   {$4,0,4,{\boldsymbol{0}_3}, 4$};
		\path[->,color=black] (32?) edge [above] node  {1} (434);
		\path[->,color=black] (3242) edge [above] node  {0} (442);
		\path[->,color=black] (322) edge [above] node  {1} (432);
		\path[->,color=black, line width=0.5mm] (324?) edge [above] node  {0} (444);
		\path[->,color=black] (3131) edge [above] node  {0} (431);
		\path[->,color=black] (3032) edge [above] node  {1} (441);
	\end{tikzpicture}
	\caption{Trellis for the $[7,4,3]$-cyclic code $\C$ defined in Example~\ref{exa:trellis-1}. To reduce clutter, we omitted states that are not connected to any final state. We also omitted two final states with $\balind=3$ as they are not reachable from the root. Also, we use colors to label the edges. Edges labelled $00$ are colored {\color{red}red}, $01$ are colored {\color{orange}orange}, $10$ are colored {\color{green}green}, and $11$ are colored {\color{blue}blue}. The path highlighted in bold corresponds to the codeword $1010110$. For more details, refer to Example~\ref{exa:trellis-1}.%\todo{HM: Sorry, I messed up the coloring. I think some colorings/labellings are incorrect.} 
	}
	\label{fig:743cyclic}
\end{figure*}
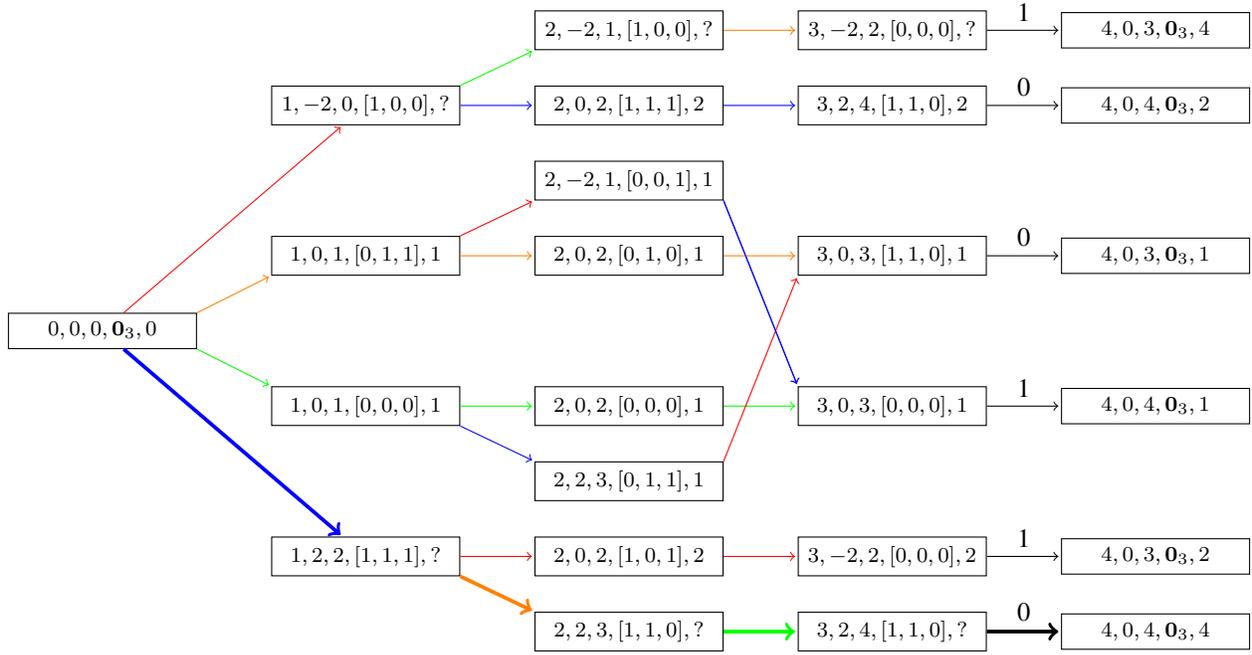

Proposition~\ref{prop:level1} initializes the values of $T$ at level~1. 
%We have 4 states here corresponding to $00,01,10$ and $11$.

\begin{proposition}\label{prop:level1}
	We have that $T(1,-2,0,\bh_1, \mathord{?}) = T(1,0,1,\bh_1+\bh_{m+1}, 1)= T(1,0,1,{\boldsymbol{0}_k},1)= T(1,2,2,\bh_{m+1}, \mathord{?})=1$.
\end{proposition}

Proposition \ref{prop:mid_level} deals with vertices $\bv$ with $2 \le \level_\bv \le m-1$.

\begin{proposition}\label{prop:mid_level}
	We have three cases that depend on the value of $\balind_\bv$.
	\begin{itemize}
		\item When $\balind_\bv = \mathord{?}$,
		\begin{align*}
			T(\level_\bv, CR_\bv, \weight_\bv, \synd_\bv,  \mathord{?}) 
			&= T(\level_\bv-1, CR_\bv+2, \weight_\bv, \synd_\bv+h_{\level_\bv-1},  \mathord{?}) \\
			&+ T(\level_\bv-1, CR_\bv, \weight_\bv-1, \synd_\bv+h_{\level_\bv-1}+ h_{\level_\bv-1+m},  \mathord{?})\\
			&+ T(\level_\bv-1, CR_\bv, \weight_\bv-1, \synd_\bv,  \mathord{?})\\
			&+ T(\level_\bv-1, CR_\bv-2, \weight_\bv-2, \synd_\bv+ h_{\level_\bv-1+m},  \mathord{?}).
		\end{align*}
		\item When $\balind_\bv \ne \mathord{?}$ and $\balind_\bv < \level_\bv$,
		\begin{align*}
			T(\level_\bv, CR_\bv, \weight_\bv, \synd_\bv,  \balind_\bv) 
			&= T(\level_\bv-1, CR_\bv+2, \weight_\bv, \synd_\bv+h_{\level_\bv-1},  \balind_\bv) \\
			&+ T(\level_\bv-1, CR_\bv, \weight_\bv-1, \synd_\bv+h_{\level_\bv-1}+ h_{\level_\bv-1+m},  \balind_\bv)\\
			&+ T(\level_\bv-1, CR_\bv, \weight_\bv-1, \synd_\bv,  \balind_\bv)\\
			&+ T(\level_\bv-1, CR_\bv-2, \weight_\bv-2, \synd_\bv+ h_{\level_\bv-1+m},  \balind_\bv).
		\end{align*}
		\item When $\balind_\bv =\level_\bv$,
		\begin{align*}
			T(\level_\bv, CR_\bv, \weight_\bv, \synd_\bv,   \mathord{?}) 
			&= T(\level_\bv-1, CR_\bv+2, \weight_\bv, \synd_\bv+h_{\level_\bv-1},   \mathord{?}) \\
			&+ T(\level_\bv-1, CR_\bv, \weight_\bv-1, \synd_\bv+h_{\level_\bv-1}+ h_{\level_\bv-1+m},  \mathord{?})\\
			&+ T(\level_\bv-1, CR_\bv, \weight_\bv-1, \synd_\bv,   \mathord{?})\\
			&+ T(\level_\bv-1, CR_\bv-2, \weight_\bv-2, \synd_\bv+ h_{\level_\bv-1+m},   \mathord{?}).
		\end{align*}
	\end{itemize}
\end{proposition}

Lastly, Proposition \ref{prop:lastlevel} considers the values $T$ at the final states. 

\begin{proposition}\label{prop:lastlevel}
	We note that there are two types of final state with weight either $m$ or $m-1$.
	\begin{itemize}
		\item The weight is $m$.
		\begin{itemize}
			\item If the index is $m$, then
			$T(m, 0, m, {\boldsymbol{0}_k},  m) = T(m-1, 2, m, h_{m},   \mathord{?}) $.
			\item If the index is $m-1$, then
			$T(m, 0, m, {\boldsymbol{0}_k},  m-1) = T(m-1, 0, m-1, h_{m},  m-1) $.
			\item If the index is $\alpha < m-1$, then
			$T(m, 0, m, {\boldsymbol{0}_k},  \alpha) = T(m-1, 2, m, h_{m},  \alpha)+T(m-1, 0, m-1, h_{m},  \alpha) $.
		\end{itemize}
		\item The weight is $m-1$.
		\begin{itemize}
			\item If the index is $m$, then
			$T(m, 0, m-1, {\boldsymbol{0}_k},  m) = T(m-1, -2, m-2, h_{m},   \mathord{?}) $.
			\item If the index is $m-1$, then
			$T(m, 0, m-1, {\boldsymbol{0}_k},  m-1) = T(m-1, 0, m-1, h_{m},  m-1) $.
			\item If the index is $\alpha < m-1$, then
			$T(m, 0, m-1, {\boldsymbol{0}_k},  \alpha) = T(m-1, 0, m-1, h_{m},  \alpha)+T(m-1, -2, m-2, h_{m},  \alpha) $.
		\end{itemize}
	\end{itemize}
\end{proposition}

%\todo{HM to continue editing.}

Finally, we have the following theorem that provides a practical method to compute the average redundancy of Scheme~C.
\begin{theorem}\label{theo:general_gammaC}
	If $\C$ is an $[n,k,d]$-cyclic code, then $\gamma_\C(i,n)$ can be computed in $O(2^{n-k}n^4)$ time.
	Hence, if $n-k$ is a constant, then $\gamma_\C(i,n)$ can be computed in $O(n^4)$ time.
\end{theorem}
\begin{proof}\label{proof:general_gammaC}
	Since each recursion formula in Propositions~\ref{prop:mid_level} and~\ref{prop:lastlevel} involve at most four summands,
	the time complexity is $O(V)$ where $V$ is the total number of states. 
	Here, we provide a simple upper bound for $V$.
	Recall that each state is labelled by five values.
	The first value corresponds to its level and this has $m+1$ possibilities.
	Since the second value, the cyclic running sum, is between $-m$ and $m$, we have $2m+1$ possibilities.
	The third value, weight, is from $0$ to $m$, while
	the fifth value, index, takes values from $1$ to $m$ and also an extra symbol $\mathord{?}$.
	In both cases, we have $m+1$ possibilities.
	Finally, the fourth value is the possible realizations of the syndrome, which is a binary vector of length $n-k$.
	Hence, there are at most $2^{n-k}$ possible syndromes.
	In summary,  $V=O(2^{n-k}n^4)$ and we have the required running time.
	%In the worst case, to count $\gamma_\C(i,2m-1)= 2*T(m,0,m,{\bf{0}_k},i)$, we need to go through all states. Hence, the time complexity is upper bound by total number of states. The range of level is $m+1$. The cyclic running sum is between $-m$ and $m$. 
	%If the cyclic running sum is $m+1$, the word at the current state has at least $m+1$ ones. Hence, the current level is at least $(m+1)/2$. To reach a final state with $CR=0$, we need to add at least $m+1$ zeros. It takes at least $(m+1)/2$ levels. As a result, the final level is $m+1$, which is a contradiction. 
	%The weight is from $0$ to $m$. The index takes values from $1$ to $m$ with an extra symbol $\mathord{?}$. The last thing is how many possible values of syndrome. Each syndrome is a vector of length $n-k$. Hence, there are at most $2^{n-k}$ possible syndrome. To sum up, the number of states is at most $2^{n-k}m^4= O(n^4)$ when $n-k$ is a constant.
\end{proof}

\begin{corollary}\label{cor:O5}
	If $\C$ is $[n,k,d]$-cyclic code, then $\rho_\C$ can be computed in $O(2^{n-k}n^5)$.
	Suppose further that $n-k$ is constant, then $\rho_\C$ can be computed in $O(n^5)$.
\end{corollary}

\section{Conclusion}

%\todo{Tu, feel free to write this. Important is to describe some open problems. Like can we get closed formulas for $\rho_\C$ for certain famous cyclic codes, like Hamming code and BCH codes.}

We studied schemes that use variable-length prefixes to encode messages onto length-$n$ codewords of weight $n/2+\sq$.  Of significance, we obtained the following results for average redundancy of our schemes.
\begin{enumerate}[(i)]
	\item When $\sq$ is a constant, we prove that our schemes uses at most $\frac{1}{2} \log{n}+2.526$ redundant bits.
	\item In particular, when $\sq = 0$, we obtained a closed form formula for the average redundancy. Moreover, in this case, the number of redundant bits asymptotically differs from the optimal redundancy by a factor of $0.2$.
	\item We extended our work to the balancing scheme with error correcting capabilities for $\sq = 0$. Although the average redundancy $\rho_C$ depends on the choice of the cyclic code $\C$ as the set of messages, we showed that $\rho_C \sim \log{n}$ when $\C = \{0,1\}^{n-1}$. 
	\item We provided polynomial-time computable formulas to estimate the average redundancy of all our schemes.
\end{enumerate}

Some open problems arise from this work.
\begin{enumerate}[(i)]
	\item {\it How to construct balancing schemes when $\sq$ grows with $n$?} This question has been considered in literature \cite{skachek2014constant}. Since our schemes inherits the simplicity of Knuth's technique, it seems unable to deal with bad words when $\sq$ is large. However, Lemmas \ref{lem:bad} and \ref{lem:badB} are elementary, and it is likely that more sophisticated techniques (like those in~\cite{skachek2014constant}) are able to balance bad words with less redundant bits. This may then reduce the average redundancy.
	\item {\it Do our schemes work with non-binary alphabet?} This question was
	prompted by a reviewer in our conference paper. Specifically, are we able to adapt our schemes and the accompanying analysis for balanced $q$-ary words (as defined in~\cite{Paluncic2018Capacity, Swart2018binary})?
	However, it is beyond the scope of this work and we defer. %Our schemes need modification to adapt with the definition of balanced weight in $q$-ary as in references \cite{Paluncic2018Capacity, Swart2018binary}. %It would be promising.
	\item {\it Can we obtain closed formulas for $\rho_\C$ when $\C$ are well-known cyclic codes?} In Section~\ref{sec:rhoC-general}, we gave a general algorithm to determine $\rho_\C$ for any given cyclic code. However, for certain classes of cyclic codes, like Hamming and BCH codes~(see for example, \cite{MS1977}), it is plausible that we can exploit certain algebraic properties to obtain closed formulas.
	%The main ingredient to count $\gamma_\C(i,n)$ is given in Theorem \ref{thm:gammaC-all}.
	%In the counting, in addition to the cyclic running sum, we need to check the syndrome to ensure messages belonging to $\C$. As we described, the transition of syndrome involving as two column of parity-check matrix $\bH$. Thus, we need to find a controllable quantity to measure the behavior of 2 columns.
	\item {\it Can the complexity $O(n^4)$ in Theorem \ref{theo:general_gammaC} be reduced?} We noted that the estimates in Theorem \ref{theo:general_gammaC} are loose since we consider all states. In fact, we can ignore many states, which are either not connected to any final state or not reachable from the root as in Figure \ref{fig:743cyclic}. Moreover, our trellis construction simply makes use any parity check matrix $\bH$. It is plausible that a column-permutation of $\bH$ may yield a smaller number of states. This is somewhat similar to the classic problem in the study of trellises (see for example, \cite{Vardy1998}). 
\end{enumerate}

\section*{Acknowledegements}
The authors would like to thank Brendan McKay and the other contributors at {\tt mathoverflow.net} for suggesting the asymptotic estimate in Proposition~\ref{prop:rhoC}.


%\bibliographystyle{IEEEtran}
%\bibliography{ref.bib}

\end{document}